\documentclass[11pt, draftclsnofoot, journal, onecolumn]{IEEEtran}

\usepackage{amsmath,stmaryrd,acronym,amssymb,amsthm,dsfont,graphicx}
\usepackage[space]{cite}

\graphicspath{{graphics/}} 

\acrodef{AEP}{Asymptotic Equipartition Property}
\acrodef{AoA}{Angle of Arrival}
\acrodef{AWGN}{Additive White Gaussian Noise}
\acrodef{BER}{Bit-Error-Rate}
\acrodef{BEC}{Binary Erasure Channel}
\acrodef{BPSK}{Binary Phase-Shift Keying}
\acrodef{BSC}{Binary Symmetric Channel}
\acrodef{CDF}[CDF]{Cumulative Distribution Function}
\acrodef{CLT}[CLT]{Central Limit Theorem}
\acrodef{CSI}[CSI]{Channel State Information}
\acrodef{DMC}[DMC]{Discrete Memoryless Channel}
\acrodef{DMS}[DMS]{Discrete Memoryless Source}
\acrodef{iid}[i.i.d.]{independent and identically distributed}
\acrodef{LPD}[LPD]{Low Probability of Detection}
\acrodef{LDPC}[LDPC]{Low-Density Parity-Check}
\acrodef{MAC}[MAC]{multiple-access channel}
\acrodef{MIMO}[MIMO]{Multiple-Input Multiple-Output} 
\acrodef{MISO}{Multiple-Input Single-Output}
\acrodef{PDF}[PDF]{Probability Distribution Function}
\acrodef{PMF}[PMF]{Probability Mass Function}
\acrodef{PPM}[PPM]{Pulse Position Modulation}
\acrodef{PSD}{Power Spectral Density}
\acrodef{QPSK}{Quadrature Phase-Shift Keying}
\acrodef{SIMO}{Single-Input Multiple-Output}
\acrodef{SNR}{Signal-to-Noise Ratio}
\acrodef{wrt}[w.r.t.]{with respect to}
\acrodef{WSS}{Wide Sense Stationary} 

\DeclareMathAlphabet{\eurm}{U}{eur}{m}{n}
\DeclareMathAlphabet{\mathbsf}{OT1}{cmss}{bx}{n}
\DeclareMathAlphabet{\mathssf}{OT1}{cmss}{m}{sl}
\DeclareMathAlphabet{\mathcsf}{OT1}{cmss}{sbc}{n}

\DeclareSymbolFont{bsfletters}{OT1}{cmss}{bx}{n}  
\DeclareSymbolFont{ssfletters}{OT1}{cmss}{m}{n}
\DeclareMathSymbol{\bsfGamma}{0}{bsfletters}{'000}
\DeclareMathSymbol{\ssfGamma}{0}{ssfletters}{'000}
\DeclareMathSymbol{\bsfDelta}{0}{bsfletters}{'001}
\DeclareMathSymbol{\ssfDelta}{0}{ssfletters}{'001}
\DeclareMathSymbol{\bsfTheta}{0}{bsfletters}{'002}
\DeclareMathSymbol{\ssfTheta}{0}{ssfletters}{'002}
\DeclareMathSymbol{\bsfLambda}{0}{bsfletters}{'003}
\DeclareMathSymbol{\ssfLambda}{0}{ssfletters}{'003}
\DeclareMathSymbol{\bsfXi}{0}{bsfletters}{'004}
\DeclareMathSymbol{\ssfXi}{0}{ssfletters}{'004}
\DeclareMathSymbol{\bsfPi}{0}{bsfletters}{'005}
\DeclareMathSymbol{\ssfPi}{0}{ssfletters}{'005}
\DeclareMathSymbol{\bsfSigma}{0}{bsfletters}{'006}
\DeclareMathSymbol{\ssfSigma}{0}{ssfletters}{'006}
\DeclareMathSymbol{\bsfUpsilon}{0}{bsfletters}{'007}
\DeclareMathSymbol{\ssfUpsilon}{0}{ssfletters}{'007}
\DeclareMathSymbol{\bsfPhi}{0}{bsfletters}{'010}
\DeclareMathSymbol{\ssfPhi}{0}{ssfletters}{'010}
\DeclareMathSymbol{\bsfPsi}{0}{bsfletters}{'011}
\DeclareMathSymbol{\ssfPsi}{0}{ssfletters}{'011}
\DeclareMathSymbol{\bsfOmega}{0}{bsfletters}{'012}
\DeclareMathSymbol{\ssfOmega}{0}{ssfletters}{'012}

\newcommand{\calA}{{\mathcal{A}}}
\newcommand{\calB}{{\mathcal{B}}}
\newcommand{\calC}{{\mathcal{C}}}
\newcommand{\calD}{{\mathcal{D}}}

\newcommand{\calF}{{\mathcal{F}}}

\newcommand{\calX}{{\mathcal{X}}}
\newcommand{\calY}{{\mathcal{Y}}}

\newcommand{\calZ}{{\mathcal{Z}}}

\newcommand{\E}[2][]{{\mathbb{E}_{#1}}{\left(#2\right)}}       
\renewcommand{\P}[2][]{{\mathbb{P}_{#1}}{\left(#2\right)}}
       
\newcommand{\avgD}[2]{{{\mathbb{D}}\!\left({#1\Vert#2}\right)}}
\newcommand{\V}[1]{{{\mathbb{V}}\!\left(#1\right)}}

\newcommand{\avgI}[1]{{{\mathbb{I}}\!\left(#1\right)}}
\newcommand{\avgH}[1]{{\mathbb{H}}\!\left(#1\right)}

\newcommand{\Hb}[1]{{\mathbb{H}_b}\left(#1\right)}

\newcommand{\card}[1]{\ensuremath{\left|{#1}\right|}}           
\newcommand{\abs}[1]{\ensuremath{\left|#1\right|}}              
\newcommand{\eqdef}{\ensuremath{\triangleq}}                    
\newcommand{\intseq}[2]{\ensuremath{\llbracket{#1},{#2}\rrbracket}}  
\newcommand{\indic}[1]{\ensuremath{\mathds{1}\!\left\{#1\right\}}}

\renewcommand{\leq}{\leqslant}
\renewcommand{\geq}{\geqslant}

\newcommand{\proddist}{%
  \mathchoice{\raisebox{1pt}{$\displaystyle\otimes$}}
             {\raisebox{1pt}{$\otimes$}}
             {\raisebox{0.5pt}{\scalebox{0.7}{$\scriptstyle\otimes$}}}
             {\raisebox{0.4pt}{\scalebox{0.6}{$\scriptscriptstyle\otimes$}}}}
\newcommand{\pn}{{\proddist n}}

\usepackage[textwidth=0.68in,textsize=footnotesize,disable]{todonotes}
\usepackage{xr,prettyref}
\presetkeys{todonotes}{color=redArcom, linecolor=redArcom,bordercolor=redArcom}{}

\newtheorem{theorem}{Theorem}
\newtheorem{remark}{Remark}

\newtheorem{lemma}{Lemma}

\newcommand{\what}[1]{\widehat{W}_{#1}} 
\newcommand{\wtilde}[1]{\widetilde{W}_{#1}}

\newcommand{\xn}{\mathbf{x}}
\newcommand{\xbarn}{\overline{\mathbf{x}}}
\newcommand{\xbar}{\overline{x}}
\newcommand{\bxn}{\mathbf{X}}
\newcommand{\yn}{\mathbf{y}}
\newcommand{\byn}{\mathbf{Y}}
\newcommand{\zn}{\mathbf{z}}
\newcommand{\bzn}{\mathbf{Z}}

\newcommand{\qbarn}[1]{\smash{\overline{Q}}_{#1}^n}
\newcommand{\qhatn}[1]{\widehat{Q}_{#1}^n}
\newcommand{\qhatj}[1]{\widehat{Q}_{j,#1}}
\newcommand{\phatj}[1]{\widehat{P}_{j,#1}}
\newcommand{\qzn}{Q_0^{\pn}}
\newcommand{\pibar}[3]{\Pi_{#1,#2,#3}}
\newcommand{\pbarcovn}[3]{\smash{\overline{P}}^{\pn}_{#1,#2,#3}}
\newcommand{\qbarcovn}[3]{\smash{\overline{Q}}^{\pn}_{#1,#2,#3}}
\newcommand{\qbarcov}[3]{\overline{Q}_{#1,#2,#3}}
\newcommand{\pbarcov}[3]{\overline{P}_{#1,#2,#3}}
\newcommand{\qo}{Q_1}
\newcommand{\qz}{Q_0}
\newcommand{\po}{P_1}
\newcommand{\pz}{P_0}
\newcommand{\pzc}[1]{P_{j,#1}^0}
\newcommand{\poc}[1]{P_{j,#1}^1}
\newcommand{\qzc}[1]{Q_{j,#1}^0}
\newcommand{\qoc}[1]{Q_{j,#1}^1}

\newcommand{\wyxn}{W_{Y|X}^{\pn}}
\newcommand{\wzxn}{W_{Z|X}^{\pn}}
\newcommand{\wyzxn}{W_{YZ|X}^{\pn}}
\newcommand{\wyx}{W_{Y|X}}
\newcommand{\wzx}{W_{Z|X}}
\newcommand{\wyzx}{W_{YZ|X}}
\newcommand{\vxx}{V_{X|\overline{X}}}

\newcommand{\pe}[1]{P_e^{(#1)}}
\newcommand{\chisquare}[2]{\chi_2\left(#1 \| #2 \right)}
\newcommand{\mun}[1]{\mu_{j,#1}^{(n)}}
\newcommand{\Psij}[1]{\Psi_{j,#1}^{(n)}}
\newcommand{\xij}[1]{\xi_{j,#1}^{(n)}}
\newcommand{\nz}{n_{j,0}}
\newcommand{\no}{n_{j,1}}
\newcommand{\rhoz}{\rho_{j,0}^{(n)}}
\newcommand{\rhoo}{\rho_{j,1}^{(n)}}

\newcommand{\mumin}{\mu_{\min}^{(n)}}
\newcommand{\agammaj}{\calA_{\gamma_j}^n}
\newcommand{\btauj}{\calB_{\tau_j}^n}
\newcommand{\lamj}{\lambda_j^{(n)}}
\renewcommand{\E}[2]{\mathbb{E}_{#1}\bracknorm{#2}}

\newcommand{\brackcurl}[1]{\left\{#1\right\}}
\newcommand{\bracknorm}[1]{\left(#1\right)}
\newcommand{\bracksq}[1]{\left[#1\right]}
\newcommand{\bigO}[1]{\mathcal{O} \left( #1 \right)}
\newcommand{\nex}{\nonumber \\}
\newcommand{\limn}{\lim_{n \to \infty}}
\newcommand{\dupspace}{\phantom{==}}
\newcommand{\etal}{\textit{et al}.}

\newcommand{\RNum}[1]{\uppercase\expandafter{\romannumeral #1\relax}}

\newtheorem{example}{Example}

\begin{document}
	
\title{Embedding Covert Information in Broadcast Communications}
\author{Keerthi Suria Kumar Arumugam, \textit{Student Member, IEEE}, and Matthieu R. Bloch, \textit{Senior Member, IEEE}\thanks{Parts of this manuscript were presented at the 2017 IEEE Information Theory Workshop~\cite{ArumugamBloch2017}. This work was supported by the National Science Foundation under Award 1527387.}\thanks{This work has been submitted to the IEEE for possible publication. Copyright may be transferred without notice, after which this version may no longer be accessible.}}
\maketitle

\begin{abstract}
	We analyze a two-receiver binary-input discrete memoryless broadcast channel, in which the transmitter communicates a common message simultaneously to both receivers and a covert message to only one of them. 
	The unintended recipient of the covert message is treated as an adversary who attempts to detect the covert transmission. 
	This model captures the problem of embedding covert messages in an innocent codebook and generalizes previous covert communication models in which the innocent behavior corresponds to the absence of communication between legitimate users. 
	We identify the exact asymptotic behavior of the number of covert bits that can be transmitted when the rate of the innocent codebook is close to the capacity of the channel to the adversary. 
	Our results also identify the dependence of the number of covert bits on the channel parameters and the characteristics of the innocent codebook.
\end{abstract}
\section{Introduction} \label{sec:introduction}
In certain scenarios, the very intention to communicate can be considered as a violation resulting in dire consequences. 
Consequently, many techniques such as spread-spectrum communications have been developed to ensure communication with \ac{LPD} also known as covert communication. 
There has been a renewed interest to study the information theoretic limits of \ac{LPD}, especially after Bash \etal~\cite{BashGoeckelTowsley2013} showed that covert communication over a point-to-point channel is subject to the \emph{square-root} law.
They showed that the transmitter can only send $\bigO{\sqrt{n}}$ bits over $n$ channel uses without being detected by the adversary. 
Several subsequent works have led to a complete characterization of the information-theoretic limits of covert communication over point-to-point classical channels. 
While the results of~\cite{WangWornellZheng2016, Bloch2016} established a tight first-order asymptotic characterization of the covert throughput over point-to-point channels,~\cite{TahmasbiBloch2017} refined the results with second-order asymptotics for different covertness metrics.
Results of~\cite{Bloch2016, CheBakshiJaggi2013} highlighted the conditions required to achieve keyless \ac{LPD} communication over \acp{DMC} and \acp{BSC}, respectively. 
Furthermore, other works have analyzed covert communication over multiple-access channels~\cite{ArumugamBloch2018}, broadcast channels~\cite{ArumugamBloch2017, Tan2017}, relay channels~\cite{Arumugam2018covertrelay, hu2018covert}, and timing channels~\cite{Soltani2015, Mukherjee2016a}. 
A few other works have analyzed scenarios in which the square-root law does not apply; for instance, scenarios in which the adversary is uncertain about the channel parameters~\cite{CheBakshiChanEtAl2014a, LeeBaxleyWeitnauerEtAl2015} or timing of the transmission~\cite{BashGoeckelTowsley2014, ArumugamBloch2016a}.
In addition, there have also been efforts to construct explicit codes for covert communication~\cite{ZhangBakshiJaggi2016, BlochGuha2017, FrecheBlochBarret2017, KadampotTahmasbiBloch2018}.

All the above works define covert communication \ac{wrt} an innocent behavior in which the transmitter does not communicate. 
In contrast, we analyze a scenario in which the innocent behavior  corresponds to the transmission of codewords from an innocent codebook that is permitted and decoded by the adversary.
The work of Dutta \etal~\cite{DuttaSahaGrunwaldEtAl2012} on covert communication using dirty constellations is one of the motivations for the present work. 
In~\cite{DuttaSahaGrunwaldEtAl2012}, the authors rely on channel noise and equipment imperfections to hide a covert signal by superimposing it on top of an innocent signal while incurring minimal distortion; consequently, the informed receiver can decode the covert message while the uninformed adversary attributes the distortion of the signal to channel impairments and hardware imperfections. 
Although the authors show that their message-hiding scheme is immune to certain statistical tests, their scheme is not fundamentally covert against a more powerful adversary.
Our objective is to develop an information-theoretic analysis of embedding covert signals in innocent communication signals while escaping detection from an adversary who is not restricted to using a small set of statistical tests.

Our model relates to several previous works. 
It can be viewed as an instance of steganography~\cite{Cachin1998}, in which part of the covertext is controlled through the design of a coding scheme and another part stems from the channel noise that is only statistically known.
The model that is closest to the one considered in this work is that of~\cite{CheBakshiChanEtAl2014}, which analyzes a broadcast setup for \acp{BSC} and exploits the additive nature of the noise in \acp{BSC}.
Tan and Lee~\cite{Tan2017} also analyzed covert communication over broadcast channels, but their channel model differs from the one considered in this work. 
In~\cite{Tan2017}, the authors study the covert capacity region for a broadcast channel model in which the transmitter simultaneously sends two different covert messages to two legitimate users while escaping detection from a third user and showed that time-division transmission is optimal. 
In contrast, our model considers two receivers, one of which tries to detect the presence of a covert message besides decoding the common message. 

We build upon the channel resolvability techniques developed in~\cite{Bloch2016, ArumugamBloch2018} for point-to-point channels and \acp{MAC}, respectively, to embed covert information into innocent transmissions.
In particular, we show that the transmitter can perturb no more than $\bigO{\sqrt{n}}$ symbols of the $n$-length sequences representing the innocent transmission to remain covert from the adversary. 
We precisely characterize the asymptotic behavior of the number of covert bits that can be transmitted when the rate of the innocent transmission approaches the capacity of the channel to the adversary. 
This characterization highlights the dependence of the number of covert bits on the channel parameters and the characteristics of the innocent codebook. 
We provide an achievability proof, a detailed converse proof, and specialize our results to a \ac{BSC}, all of which were omitted in~\cite{ArumugamBloch2017}. 

The remainder of the paper is organized as follows. 
In Section~\ref{sec:notation}, we set the notation used in the paper, and in Section~\ref{sec:channelmodel}, we formally introduce our channel model. 
In Section~\ref{sec:preliminaries}, we develop a preliminary result that captures the essence of our approach to embedding covert information in innocent transmissions. 
Finally, we present our main result in Section~\ref{sec:mainresult}, which consists in an achievability and a converse characterizing the optimal asymptotic number of reliable and covert bits. 

\section{Notation} \label{sec:notation}
We denote random variables and their realizations in upper and lower case, respectively.
All sequences in boldface are $n$-length sequences, where $n \in \mathbb{N}^*$, unless specified otherwise. 
A sequence of random variables $\bracknorm{Y_j, Y_{j+1}, \ldots, Y_k}$ is denoted by $\byn_j^k$. 
The element at position $\ell \in \intseq{1}{n}$ of a sequence $\xn_{j}$ is denoted by $x_{j,\ell}$.
We interpret $\log$ and $\exp$ to the base $e$; the results can be interpreted in bits by converting $\log$ to the base $2$. 
Adhering to standard information-theoretic notation, $\avgH{X}$ and $\avgI{X;Y}$ represent the average entropy of $X$ and the average mutual information between $X$ and $Y$, respectively.
If the distribution of $X$ is $P$ and the channel between $X$ and $Y$ is $W_{Y|X}$, then $\mathbb{I}\bracknorm{P, W_{Y|X}}$ also represents the average mutual information between $X$ and $Y$. 
For $x \in \bracksq{0,1} $, $\Hb{x}$ denotes the average binary entropy of $x$. 
For two distributions $P$ and $Q$ on the same finite alphabet $\calX$, the Kullback-Liebler (KL) divergence is $\avgD{P}{Q} \eqdef \smash{ \sum_x P(x) \log \frac{P(x)}{Q(x)} }$, the variational distance is $\V{P,Q} \eqdef \frac{1}{2} \sum_x \abs{P(x) - Q(x)}$, and the chi-squared distance is $\chisquare{P}{Q} \eqdef \smash{\sum_x \frac{\bracknorm{P(x) - Q(x)}^2}{Q(x)}}$. 
Pinsker's inequality states that $\V{P,Q}^2 \leq \frac{1}{2} \avgD{P}{Q}$. 
If $P$ is absolutely continuous \ac{wrt} $Q$, we write $P \ll Q$.

\section{Channel model} \label{sec:channelmodel}
\begin{figure}[b] 
  \centering
  \includegraphics[width=0.7\linewidth]{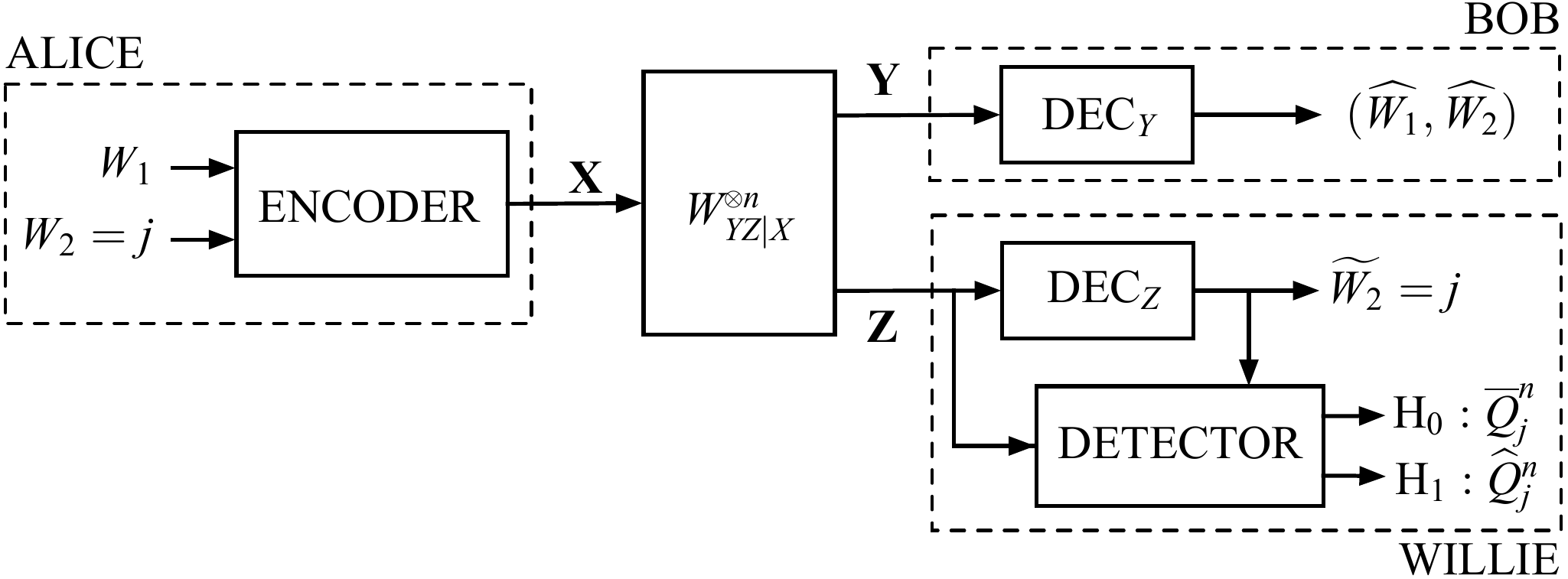} 
  \caption{Model of covert communication over a discrete memoryless broadcast channel for a fixed common message $W_2 = j$.}
  \label{fig:channelmodel}
\end{figure}
We analyze a channel model in which Alice, the transmitter, communicates a common message to both Bob, the receiver, and Willie, the warden, and a covert message to Bob alone over a discrete memoryless broadcast channel $\bracknorm{\calX, W_{YZ|X}, \calY, \calZ}$. 
We assume that the transmitter uses a binary input alphabet $\calX \eqdef \brackcurl{0,1}$ and that the output alphabets $\calY$ and $\calZ$ are finite. 
Furthermore, we assume that all terminals are synchronized and possess complete knowledge of the coding scheme used.

As illustrated in Figure~\ref{fig:channelmodel}, Alice wishes to communicate a uniformly distributed common message $W_2 \in \intseq{1}{M_2}$ to both Bob and Willie, and a uniformly distributed covert message $W_1 \in \intseq{1}{M_1}$ to Bob alone. 
Alice may also choose not to transmit any covert message, in which case, she sets $W_1 = 0$. 
She then encodes the message pair $\bracknorm{W_1, W_2} = (i,j)$ into an $n$-length codeword $\bxn_{ij} $.  
We label the collection of codewords $\brackcurl{\bxn_{0j}}_{j=1}^{M_2}$ as the \emph{innocent codebook}. 
Alice sends the codeword over the discrete memoryless broadcast channel in $n$ channel uses, at the end of which, Bob and Willie observe the $n$-length sequences $\byn$ and $\bzn$, respectively. 
Since the channel is memoryless, we denote the transition probability corresponding to $n$ uses of the channel by $ \wyzxn \eqdef \prod_{i=1}^n \wyzx $. 
For $a \in \calX$, we denote the output distributions induced by each input symbol at Bob and Willie by $P_a(y) \eqdef \wyx(y|a)$ and $Q_a(z) \eqdef \wzx(z|a)$, respectively. 
For $a,b \in \calX$ with $a \neq b$, we assume $P_a \ll P_b$, $Q_a \ll Q_b$, $Q_a \neq Q_b$. 
Without the first assumption, Bob has an unfair advantage over Willie~\cite{Bloch2016}. 
Without the second and third assumptions, achieving covert communication becomes either impossible or trivial~\cite{WangWornellZheng2016, Bloch2016}. 
We also make the following assumptions. 
\begin{itemize}
	\item The channel $\bracknorm{\calX, \wzx, \calZ}$ to Willie admits a unique capacity-achieving input distribution $\Lambda$, for which $\Lambda(1) \eqdef \lambda^*$, where $\lambda^* > 0$. Many channels encountered in practice satisfy this assumption.\footnote{Note that this assumption is required to prove only the converse.}
	\item $\mathbb{I}\bracknorm{\Lambda, \wyx} \geq \mathbb{I}\bracknorm{\Lambda, \wzx}$, so that Willie limits the rate of the common message. 
\end{itemize}
 
Upon observing the noisy sequence $\bzn$, Willie forms an estimate $\wtilde{2}$ of $W_2$. 
We measure reliability at Willie using the following metric,
\begin{align}
	\pe{2} \eqdef \frac{1}{M_2} \sum_{j=1}^{M_2} P_{e,j}^{(2)} \eqdef \mathbb{E}_{W_2}\bracknorm{P_{e,W_2}^{(2)}} , \label{eq:1}
\end{align}
where
\begin{align}
	P_{e,j}^{(2)} \eqdef  \P{\wtilde{2} \neq j | W_1 = 0, W_2 = j}
	+ \P{\wtilde{2} \neq j | W_1 \neq 0, W_2 = j}. 
\end{align}
Strictly speaking, $\smash{\pe{2}}$ is not a probability measure; however, a small $\pe{2}$ ensures that the average decoding error probability of the common message at Willie is small regardless of the presence of $W_1$. 
Note that there is no prior on whether $W_1 = 0$ or $W_1 \neq 0$. 
Willie attempts to detect the presence of a non-zero covert message by performing a binary-hypothesis test on his observation $\bzn$ to distinguish between the hypotheses $H_0 \eqdef \brackcurl{W_1 = 0}$ and $H_1 \eqdef \brackcurl{W_1 \neq 0}$.
We denote Willie's Type I and Type II errors by $\alpha$ and $\beta$, respectively. 
For a fixed $W_2 = j$, the output distribution observed by Willie is 
\begin{align}
	\qbarn{j}(\zn) & \eqdef \wzxn\bracknorm{\zn|\xn_{0j}}, \quad \text{if } W_1 = 0, \label{eq:2} \\
	\qhatn{j}(\zn) & \eqdef \frac{1}{M_1} \sum_{i=1}^{M_1} \wzxn\bracknorm{\zn|\xn_{ij}}, \quad \text{else}. \label{eq:3}
\end{align}
For a fixed common message $W_2 = j$, we measure the covertness of $W_1$ by the KL divergence $\smash{\avgD{\qhatn{j}}{\qbarn{j}}}$ since any statistical test~\cite{LehmannRomano2014} conducted on $\bzn$ by Willie must satisfy $\alpha + \beta \geq 1 - \sqrt{\avgD{\qhatn{j}}{\qbarn{j}}}$.  
A vanishing KL divergence ensures that $\alpha + \beta = 1$ in the limit, 
so that Willie's statistical test is no better than a random guess making the test futile in detecting the presence of a covert message. 

Upon observing $\byn$, Bob forms an estimate $\smash{\bracknorm{\smash{\what{1}, \what{2}}}}$ of the transmitted message pair $\bracknorm{W_1, W_2}$. We measure reliability at Bob using the metric
\begin{align}
	P_e^{(1)} & \eqdef  \frac{1}{M_2} \sum_{j=1}^{M_2} \bracknorm{ P_{e,1,j}^{(1)} + P_{e,2,j}^{(1)}} ,  \\
	& \eqdef \mathbb{E}_{W_2}\bracknorm{ P_{e,1,W_2}^{(1)}} +  \mathbb{E}_{W_2}\bracknorm{P_{e,2,W_2}^{(1)}} \label{eq:4}, 
\end{align}
where
\begin{align}
	  P_{e,1,j}^{(1)} & \eqdef \P{\widehat{W}_2 \neq j|W_1=0, W_2 = j} + \P{\widehat{W}_2\neq j|W_1\neq 0, W_2 = j}, \label{eq:4a} \\
	  P_{e,2,j}^{(1)} & \eqdef \P{\widehat{W}_1 \neq 0 | W_1 = 0, \widehat{W}_2=W_2=j}  + \P{\widehat{W_1}\neq W_1|W_1 \neq 0, \widehat{W}_2=W_2=j}. \label{eq:4b}
\end{align}
Despite $\pe{1}$ not being an error probability in the strict sense, a small $\pe{1}$ guarantees that the average error probability of the covert message and the common message at Bob is small. 

Our main objective is to characterize the optimal scaling of $\log M_1$ and $\log M_2$ with $n$ such that 
\begin{align}
	& \limn \pe{1} = \limn \pe{2} = 0, \label{eq:5} \\
	\forall j \in \intseq{1}{M_2}, \quad & \limn \avgD{\qhatn{j}}{\qbarn{j}} = 0. \label{eq:6}
\end{align}
Note that we choose to satisfy the more stringent requirement that $\limn \smash{\avgD{\qhatn{j}}{\qbarn{j}}}$ vanishes for every $j \in \intseq{1}{M_2}$ so that the hypothesis test used by Willie is futile in detecting the presence of any covert message for \emph{every choice} of the common message and not just on average. 

A couple of comments are now in order. 
First, note that our goal is twofold here: we wish to design a reliable code to communicate a common message and a reliable code to embed a covert message; this is a \emph{joint} code-design problem, and we do not address the problem of embedding covert bits into a fixed code for the common message.
Second, the problem generalizes previous works on covert communication, in which covertness was measured \ac{wrt} the innocent distribution $\qzn$ corresponding to the transmission of the all-zero sequence. 
In our case, for $W_2 = j \in \intseq{1}{M_2}$, covertness is measured \ac{wrt} the distribution $\qbarn{j}$, which is a product distribution that is not identically distributed and corresponds to the communication of the innocent codeword mapped to the common message $W_2 = j$. 

\section{Preliminaries} \label{sec:preliminaries}
Following the approach put forward in~\cite{Bloch2016}, we define a \emph{covert stochastic process}, which serves as the target distribution that our covert code approximates. 
By introducing the covert process, we precisely quantify the fraction of symbols in the innocent codeword that Alice can perturb to transmit covert information while simultaneously avoiding detection by Willie. 
\begin{figure}
  \centering
  \includegraphics[width=0.7\linewidth]{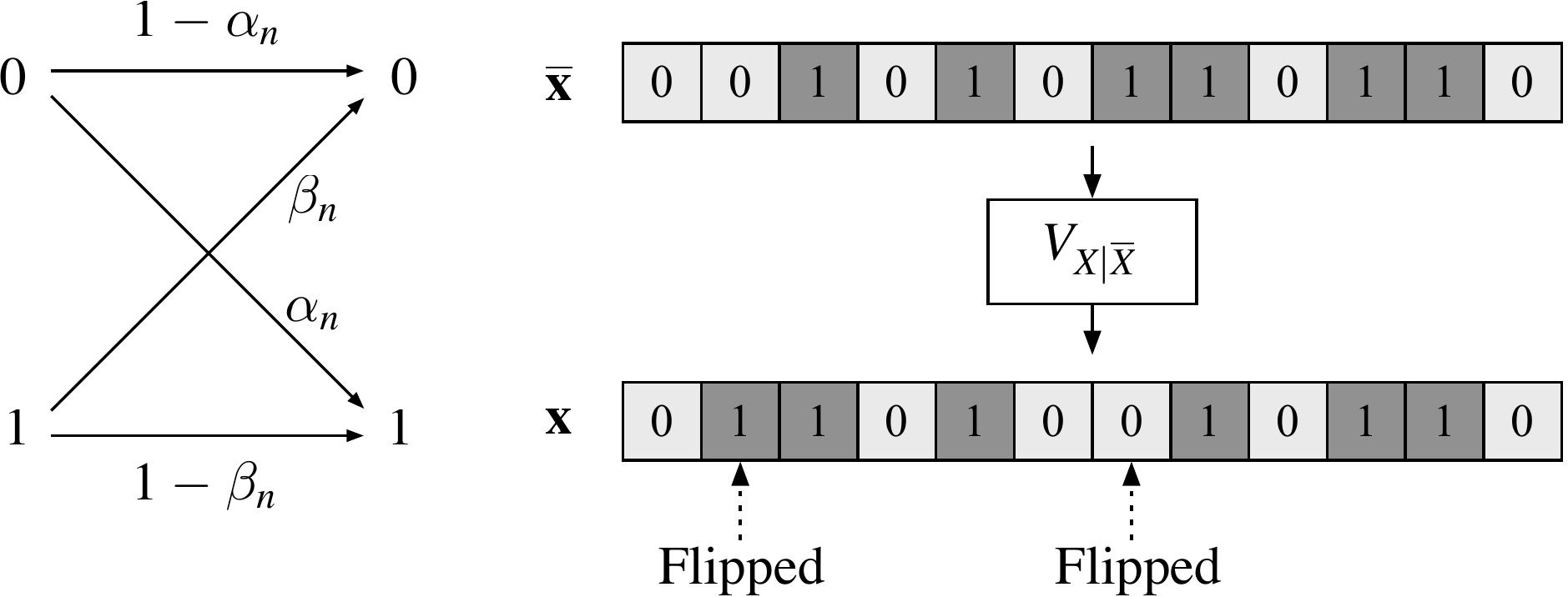}
  \caption{Binary asymmetric channel $\vxx$ and an illustration of  innocent symbols flipped by the channel $\vxx$.}
  \label{fig:covproc}
\end{figure}
For a fixed $n \in \mathbb{N}^*$ and a sequence $\xbarn \in \calX^n$, we define the covert process as the output of the binary asymmetric channel $\vxx$ illustrated in Figure~\ref{fig:covproc} such that $\vxx(1|0) \eqdef \alpha_n$ and $\vxx(0|1) \eqdef \beta_n$, where $\alpha_n, \beta_n \in (0,1)$ are cross-over probabilities. 
We denote the distribution of the covert process by $\pibar{\xbarn}{\alpha_n}{\beta_n}$ defined as
\begin{align}
	\pibar{\xbarn}{\alpha_n}{\beta_n}(\xn) \eqdef \prod_{i=1}^n \vxx \bracknorm{x_i|\xbar_i}. \label{eq:8}
\end{align}
We set $\gamma_n \eqdef \frac{\beta_n}{\alpha_n}$, and when defining a sequence $\brackcurl{\gamma_n}_{n \in \mathbb{N}^*}$, we ask that it converges to $\gamma \in \mathbb{R}^+$. 
The transmission of the covert process through the \acp{DMC} $\bracknorm{\calX, \wyx, \calY}$ and $\bracknorm{\calX, \wzx, \calZ}$ induces the output distributions
\begin{align}
	\pbarcovn{\xbarn}{\alpha_n}{\beta_n}(\yn) \eqdef \sum_{\xn} \wyxn \bracknorm{\yn|\xn} \pibar{\xbarn}{\alpha_n}{\beta_n}(\xn), \label{eq:9} \\
	\qbarcovn{\xbarn}{\alpha_n}{\beta_n}(\zn) \eqdef \sum_{\xn} \wzxn \bracknorm{\zn|\xn} \pibar{\xbarn}{\alpha_n}{\beta_n}(\xn), \label{eq:10}
\end{align}
at Bob and Willie, respectively.  
Note that both $\pbarcovn{\xbarn}{\alpha_n}{\beta_n}$ and $\qbarcovn{\xbarn}{\alpha_n}{\beta_n}$ are product distributions, and setting both $\alpha_n$ and $\beta_n$ to $0$ results in a distribution   $\qbarcovn{\xbarn}{0}{0}$ at Willie. 
We now have the following generalization of~\cite[Lemma 1]{Bloch2016}. 
\begin{lemma} \label{lem:covertprocess}
	Let $\alpha_n, \beta_n \in (0,1)$ be such that $\limn \alpha_n = \limn \beta_n = 0$. 
	Let ${ \normalfont \xbarn} \in \calX^n$ and set $\lambda_n \eqdef \sum_{i=1}^n \frac{\indic{ \xbar_i = 1}}{n}$. 
	Then, for a large $n \in \mathbb{N}^*$, we bound ${\normalfont \avgD{\qbarcovn{\xbarn}{\alpha_n}{\beta_n}}{\qbarcovn{\xbarn}{0}{0}}}$ by
	\begin{align}
		n \bracknorm{\bracknorm{1 - \lambda_n} \frac{\alpha_n^2}{2}\bracknorm{1 + \sqrt{\alpha_n}} \chisquare{\qo}{\qz} + \lambda_n \frac{\beta_n^2}{2}\bracknorm{1 + \sqrt{\beta_n}} \chisquare{\qz}{\qo} } \geq \avgD{\qbarcovn{\xbarn}{\alpha_n}{\beta_n}}{\qbarcovn{\xbarn}{0}{0}} \nex
		\geq n \bracknorm{\bracknorm{1 - \lambda_n} \frac{\alpha_n^2}{2}\bracknorm{1 - \sqrt{\alpha_n}} \chisquare{\qo}{\qz} + \lambda_n \frac{\beta_n^2}{2}\bracknorm{1 - \sqrt{\beta_n}} \chisquare{\qz}{\qo} }. \label{eq:cov1}
	\end{align}
\end{lemma}
The proof of Lemma~\ref{lem:covertprocess} is provided in Appendix~\ref{sec:covprf}. 
For any $\xbarn$, upon choosing sequences $\brackcurl{\alpha_n}_{n \in \mathbb{N^*}}$ and $\brackcurl{\beta_n}_{n \in \mathbb{N^*}}$, such that $\limn n \alpha_n^2 = \limn n \beta_n^2 = 0 $, we obtain 
\begin{align}
	\limn \avgD{\qbarcovn{\xbarn}{\alpha_n}{\beta_n}}{\qbarcovn{\xbarn}{0}{0}} = 0, \label{eq:10_a}
\end{align}
which shows that $\qbarcovn{\xbarn}{\alpha_n}{\beta_n}$ is indistinguishable from $\qbarcovn{\xbarn}{0}{0}$ at Willie when the fraction of flips is small enough. 
In addition, it is also possible to choose $\brackcurl{\alpha_n}_{n \in \mathbb{N^*}}$ and $\brackcurl{\beta_n}_{n \in \mathbb{N^*}}$ such that $\limn n \alpha_n = \limn n \beta_n = \infty $ to flip an infinite number of innocent symbols as $n \to \infty$, while still ensuring that $\qbarcovn{\xbarn}{\alpha_n}{\beta_n}$ is indistinguishable from $\qbarcovn{\xbarn}{0}{0}$ according to~\eqref{eq:cov1}.
If we set $\xbarn = \xn_{0j}$, where $\xn_{0j}$ is the innocent codeword corresponding to $W_2 = j$, the distribution $\qbarcovn{\xn_{0j}}{0}{0}$ is the innocent distribution corresponding to $W_2 = j$ and is equivalent to the distribution $\qbarn{j}$ in~\eqref{eq:2}. 

\section{Main result} \label{sec:mainresult}
We now characterize the exact scaling of the number of covert bits when the common message is transmitted at a rate approaching the capacity of the channel to Willie. 
For the transmission of covert bits without a secret key, Bob is required to possess a certain advantage over Willie, which we precisely characterize in the following theorem. 
\begin{theorem} \label{thm:ach}
	For the channel model described in Section~\ref{sec:channelmodel}, if there exists $\gamma \geq 0$ such that 
	\begin{align}
		\bracknorm{1 - \lambda^*} \avgD{\po}{\pz} + \lambda^* \gamma \avgD{\pz}{\po} > \bracknorm{1 - \lambda^*} \avgD{\qo}{\qz} + \lambda^* \gamma \avgD{\qz}{\qo}, \label{eq:ach1}
	\end{align}
	there exist keyless covert communication schemes such that
	\begin{align}
		\limn \frac{\log M_2}{n} & = \mathbb{I}\bracknorm{\Lambda,\wzx}, \label{eq:ach2}
	\end{align}
	and for all $j \in \intseq{1}{M_2}$, 
	\begin{align}
		&\!\!\!\!\!\!\limn \frac{\log M_1}{\sqrt{n \avgD{\qhatn{j}}{\qbarn{j}}}}  = \max_{\gamma \geq 0} \frac{\sqrt{2} \bracknorm{ \bracknorm{1 - \lambda^*} \avgD{\po}{\pz} + \lambda^* \gamma \avgD{\pz}{\po}}}{\sqrt{\bracknorm{1 - \lambda^*} \chisquare{\qo}{\qz} + \lambda^* \gamma^2 \chisquare{\qz}{\qo}}}, \label{eq:ach3} \\
		&\!\!\!\!\!\! \limn P_e^{(1)} = \limn P_e^{(2)} = \limn \avgD{\qhatn{j}}{\qbarn{j}} = 0. \label{eq:ach4}
	\end{align}
\end{theorem}
\begin{proof}
We first show that Bob can decode the covert message, and both Bob and Willie can decode the common message reliably. 
Using channel resolvability techniques, we then show that the induced distribution $\smash{\qhatn{W_2}}$ corresponding to the common message $W_2$ is indistinguishable from the covert stochastic process $\smash{\qbarcovn{\xn_{0W_2}}{\alpha_n}{\beta_n}}$ when averaged over all choices of the common message $W_2$. 
Finally, we identify a coding scheme that achieves~\eqref{eq:ach2} and~\eqref{eq:ach3} such that~\eqref{eq:ach1} and~\eqref{eq:ach4} are satisfied. 
	\paragraph{Random code generation}
	Define a set $\calD_\epsilon^n \eqdef \brackcurl{  \xn:  \abs{ \smash{ \frac{\textnormal{wt}\bracknorm{\xn}}{n}} - \lambda^* } < \epsilon }$, where $\textnormal{wt}\bracknorm{\xn} \eqdef \card{ \ell \in \intseq{1}{n}: x_{\ell} = 1}$ is the weight of $\xn$. 
	For $j \in \intseq{1}{M_2}$, we generate $M_2$ codewords $\smash{\xn_{0j} \in \calX^n}$ independently at random according to the distribution $P_X^n$ defined by
	\begin{align}
		P_X^{n}\bracknorm{\xn} \eqdef \frac{\Lambda^\pn\bracknorm{\xn} \indic{\xn \in \calD_\epsilon^n} }{\P[\Lambda]{\bxn \in \calD_\epsilon^n}}
	\end{align} 
	Generating $\brackcurl{\xn_{0j}}_{j=1}^{M_2}$ according to $P_X^n$ ensures that every $\xn_{0j}$ is $\epsilon$-letter typical \ac{wrt} the distribution $\Lambda$.  
	We label this set of $M_2$ codewords as the innocent codebook $\calC_2$. 
	For every $W_2 = j \in \intseq{1}{M_2}$, we generate $M_1$ codewords independently at random according to the distribution $\pibar{\xn_{0j}}{\alpha_n}{\beta_n}$ and label this set of codewords as the covert sub-codebook  $\calC_{1,j}$ corresponding to the common message $W_2 = j$. 
	Alice encodes the message pair $\bracknorm{W_1, W_2} = \bracknorm{i,j}$, where $i \in \intseq{1}{M_1}$ and $j \in \intseq{1}{M_2}$, to the codeword $\xn_{ij} \in \calC_{1,j}$ and transmits it through the discrete memoryless broadcast channel. 
	Defining
	\begin{align}
		W_{Y|\overline{X}}(y|\xbar) & \eqdef \sum_{x} \wyx(y|x)\vxx(x|\xbar), \\
		W_{Z|\overline{X}}(z|\xbar)  & \eqdef \sum_{x} \wzx(z|x)\vxx(x|\xbar), 
	\end{align} 
	we show that the decoding error probability of the common message at Bob and Willie averaged over all random codebooks $\calC$ decays exponentially in the following lemma. 
	\begin{lemma} \label{lem:commonmsg}
		For any $\mu \in \bracknorm{0,1}$ and $n$ large enough, and 
		\begin{align}
			\log M_2 & <  \bracknorm{1 - \mu} n \avgI{\Lambda, W_{Y|\overline{X}}}, \label{eq:72a} \\
			\log M_2 & <  \bracknorm{1 - \mu} n \avgI{\Lambda, W_{Z|\overline{X}}}, \label{eq:72b}
		\end{align}
		we have
		\begin{align}
			\E{\calC}{ \mathbb{E}_{W_2}{P_{e,1, W_2}^{\bracknorm{1}}}} & \leq \exp \bracknorm{-\xi_1 n }, \label{eq:72c} \\
			\E{\calC}{ \mathbb{E}_{W_2}{P_{e,W_2}^{\bracknorm{2}}}} & \leq \exp \bracknorm{-\xi_1 n }, \label{eq:72d}
		\end{align}
		for an appropriate constant $\xi_1 > 0$. 
	\end{lemma}
	The proof of Lemma~\ref{lem:commonmsg} follows the random coding argument outlined in~\cite[Section 7.3]{gallager1968information} and is omitted here. 
	Note that $\forall (\xbar, z) \in \calX \times \calZ $, 
	\begin{align}
		W_{Z|\overline{X}}(z|\xbar) = \vxx(0|\xbar) \qz(z) + \vxx(1|\xbar) \qo(z). \label{eq:73}
	\end{align}
	Consequently, we have
	\begin{align}
		W_{Z|\overline{X}}(z|0) & = \qz(z) + \alpha_n \bracknorm{ \qo(z) -\qz(z) } , \label{eq:73a} \\
		W_{Z|\overline{X}}(z|1) & = \qo(z) + \beta_n \bracknorm{\qz(z) - \qo(z)} . \label{eq:73b}
	\end{align} 
	Since $\limn \alpha_n = \limn \beta_n = 0$, for a large $n$, the channels $W_{Z|\overline{X}}$ and $\wzx$ are identical in the limit of large blocklength. 
	The same argument extends to channels $W_{Y|\overline{X}}$ and $\wyx$. 
	Defining $\smash{Q_Z(z) \eqdef \sum_{\xbar} \Lambda(\xbar) W_{Z|\overline{X}}(z|\xbar)}$ and expanding the mutual information term in~\eqref{eq:72b} using~\eqref{eq:73a} and~\eqref{eq:73b}, we obtain
	\begin{align}
		\avgI{\Lambda, W_{Z|\overline{X}}} & = \sum_z \Bigg( \bracknorm{1 - \lambda^*} \bracknorm{\qz(z) + \alpha_n \bracknorm{ \qo(z) - \qz(z)}} \log \bracknorm{\frac{\qz(z) + \alpha_n \bracknorm{\qo(z) - \qz(z)}}{Q_Z(z)}} \nex
		& \dupspace \dupspace + \lambda^* \bracknorm{\qo(z) + \beta_n \bracknorm{\qz(z) - \qo(z) } } \log \bracknorm{\frac{\qo(z) + \beta_n \bracknorm{\qz(z) - \qo(z)}}{Q_Z(z)}} \Bigg) \displaybreak \\
		& = \sum_z \bracknorm{ \bracknorm{1 - \lambda^*} \qz(z) \log \frac{\qz(z)}{Q_Z(z)}  + \lambda^* \qo(z) \log \frac{\qo(z)}{Q_Z(z)} } +\bigO{\alpha_n} + \bigO{\beta_n} \\
		& = \avgI{\Lambda, \wzx} +\bigO{\alpha_n} + \bigO{\beta_n}. \label{eq:72e}
	\end{align}
  	Similarly, we obtain
	\begin{align}
		\avgI{\Lambda, W_{Y|\overline{X}}} & = \avgI{\Lambda, \wyx} +\bigO{\alpha_n} + \bigO{\beta_n}. \label{eq:72f}
	\end{align} 
	Combining~\eqref{eq:72a},~\eqref{eq:72b},~\eqref{eq:72e}, and~\eqref{eq:72f}, we obtain
	\begin{align}
		\frac{\log M_2}{n} & < \bracknorm{1 - \mu} \avgI{\Lambda, \wyx} +\bigO{\alpha_n} + \bigO{\beta_n}, \label{eq:72g} \\
		\frac{\log M_2}{n} & < \bracknorm{1 - \mu} \avgI{\Lambda, \wzx} +\bigO{\alpha_n} + \bigO{\beta_n}. \label{eq:72h}
	\end{align}
	Our assumption that $\mathbb{I}\bracknorm{\Lambda, \wyx} \geq \mathbb{I}\bracknorm{\Lambda, \wzx}$ and~\eqref{eq:72h} render~\eqref{eq:72g} unnecessary.  
	Hence, the average decoding error probability of the common message at both Bob and Willie vanishes in the limit of large blocklength if
	\begin{align}
		\limn \frac{\log M_2}{n} = \bracknorm{1 - \xi} \avgI{\Lambda, \wzx}, 
	\end{align}
	for an arbitrary $\xi > 0$. 
	Henceforth, we assume that both Bob and Willie have decoded the common message successfully. 
	\paragraph{Channel reliability analysis}
	We now prove that the decoding error probability of the covert message at Bob decays exponentially. 
	For $i \in \intseq{1}{M_1}$, the following events lead to a decoding error at Bob,
	\begin{itemize}
		\item codeword $\xn_{0j}$ is transmitted, and the decoder incorrectly estimates $\widehat{W}_1 = i$,
		\item codeword $\xn_{ij}$ is transmitted, and the decoder incorrectly estimates $\widehat{W}_1 = 0$,
		\item codeword $\xn_{ij}$ is transmitted, and the decoder incorrectly estimates $\widehat{W}_1 = i' \in \intseq{1}{M_1}$, where $i' \neq i$. 
	\end{itemize}
	The decoding error probability of the covert message at Bob averaged over all random codebooks satisfies the following lemma. 
	\begin{lemma} \label{lem:rel}
		For any $\mu \in (0,1)$, an $n$ large enough, and
		\begin{align}
			\!\!\log M_1 \! = \!\bracknorm{1 - \mu} \!n \!\bracknorm{\bracknorm{1 - \lambda^*}\alpha_n \avgD{\po}{\pz} \!+\! \lambda^* \beta_n \avgD{\pz}{\po} }\!, \label{eq:77}
		\end{align}
		we have
		\begin{align}
			\E{\calC}{ \mathbb{E}_{W_2} P_{e,2,W_2}^{(1)}} \leq \exp \bracknorm{-\xi_2 n \alpha_n} +  \exp \bracknorm{-\xi_2 n \beta_n}, \label{eq:78}
		\end{align}
		for an appropriate $\xi_2 >0$.
	\end{lemma}
	The proof of Lemma~\ref{lem:rel} is provided in Appendix~\ref{sec:relprf}. 
	\paragraph{Channel resolvability analysis}
	We now show that the KL divergence between the induced distribution $\qhatn{W_2}$ and the covert stochastic process $\smash{\qbarcovn{\bxn_{0W_2}}{\alpha_n}{\beta_n}}$ averaged over all choices of the common message and all random codebooks vanishes in the limit of large blocklength. 
	\begin{lemma} \label{lem:res}
		For any $\nu > 0$, an $n$ large enough, and 
		\begin{align}
			\!\log M_1 \!=\! \bracknorm{1 + \nu} \!n\! \bracknorm{\bracknorm{1 - \lambda^*}\alpha_n \avgD{\qo}{\qz} \!+\! \lambda^* \! \beta_n \avgD{\qz}{\qo} }\!, \label{eq:81}
		\end{align}
		we have
		\begin{align}
			{\normalfont \E{\calC\!\!}{\mathbb{E}_{W_2} \avgD{\!\qhatn{W_2}}{\qbarcovn{\bxn_{0W_2}}{\alpha_n}{\beta_n\!}}\!}} \!\leq\! \exp\bracknorm{-\xi_3 n \alpha_n} + \exp \bracknorm{-\xi_3 n \beta_n}, \label{eq:82}			
		\end{align}
		for an appropriate $\xi_3 > 0$. 
	\end{lemma}
	The proof of Lemma~\ref{lem:res} is provided in Appendix~\ref{sec:resprf}. 
	\paragraph{Identification of a specific code}	 
	Using Markov's inequality, we obtain
	\begin{multline}
	\!\!\!\!\!\!\!\!\mathbb{P} \Bigg( \mathbb{E}_{W_2}{P_{e,1, W_2}^{\bracknorm{1}}} < 8 \E{\calC}{ \mathbb{E}_{W_2}{P_{e,1, W_2}^{\bracknorm{1}}}} \bigcap \mathbb{E}_{W_2}{P_{e,2, W_2}^{\bracknorm{1}}} < 8 \E{\calC}{ \mathbb{E}_{W_2}{P_{e,2, W_2}^{\bracknorm{1}}}} \bigcap \mathbb{E}_{W_2}{P_{e,W_2}^{\bracknorm{2}}} <8 \E{\calC}{ \mathbb{E}_{W_2}{P_{e,W_2}^{\bracknorm{2}}}} \\
	 \bigcap {\mathbb{E}_{W_2} \avgD{\!\qhatn{W_2}}{\qbarcovn{\xn_{0W_2}}{\alpha_n}{\beta_n\!}}\!} < 8 \E{\calC\!\!}{\mathbb{E}_{W_2} \avgD{\!\qhatn{W_2}}{\qbarcovn{\bxn_{0W_2}}{\alpha_n}{\beta_n\!}}\!} \Bigg) \geq \frac{1}{2}. \label{eq:95}
	\end{multline}
  Defining $\epsilon_n \eqdef \exp \bracknorm{-\xi_4 n \alpha_n} +  \exp \bracknorm{-\xi_4 n \beta_n}$ for an appropriate constant $\xi_4 > 0$, we conclude from~\eqref{eq:95} that there exists at least one coding scheme $\calC^*$ such that for a large $n$,
	\begin{align}
		\mathbb{E}_{W_2}{P_{e,1, W_2}^{\bracknorm{1}}} & \leq \epsilon_n, \label{eq:96} \\
		\mathbb{E}_{W_2}{P_{e,2, W_2}^{\bracknorm{1}}} & \leq \epsilon_n, \label{eq:97} \\
		\mathbb{E}_{W_2}{P_{e,W_2}^{\bracknorm{2}}} & \leq \epsilon_n, \label{eq:97a} \\
		{\mathbb{E}_{W_2} \avgD{\qhatn{W_2}}{\qbarcovn{\xn_{0W_2}}{\alpha_n}{\beta_n}}} & \leq \epsilon_n, \label{eq:97b}
	\end{align}
	where $\xn_{0W_2} \in \calC_2^*$ is the codeword corresponding to the common message $W_2$. 
	We expurgate half of the innocent codewords and their corresponding covert sub-codebooks such that for every remaining $W_2 = j$, we have
	\begin{align}
		{P_{e,1, j}^{\bracknorm{1}}} & \leq 8\epsilon_n, \label{eq:97c} \\
		{P_{e,2, j}^{\bracknorm{1}}} & \leq 8\epsilon_n, \label{eq:97d} \\
		{P_{e,j}^{\bracknorm{2}}} & \leq 8\epsilon_n, \label{eq:97e} \\
		\avgD{\qhatn{j}}{\qbarcovn{\xn_{0j}}{\alpha_n}{\beta_n}} & \leq 8\epsilon_n, \label{eq:97f}
	\end{align}
	without affecting the asymptotic rate of the common message.
	Note that covertness is not affected by expurgating whole covert sub-codebooks. 
	Since $\limn \epsilon_n = 0$,~\eqref{eq:97c},~\eqref{eq:97d}, and~\eqref{eq:97e} imply $\limn P_e^{(1)} = \limn P_e^{(2)} = 0$. 
	
	Using Pinsker's inequality with~\eqref{eq:97f} ensures that $\V{\qhatn{j},\qbarcovn{\xn_{0j}}{\alpha_n}{\beta_n}} \leq \exp \bracknorm{-\xi_5 n \alpha_n} + \exp \bracknorm{-\xi_5 n \beta_n} $ for an appropriate $\xi_5 > 0$. 
	Then, we write
	\begin{align}
		\!\!\avgD{\qhatn{j}}{\qbarn{j}\!}\! & =\! \avgD{\qhatn{j}}{\qbarcovn{\xn_{0j}}{\alpha_n}{\beta_n}} + \avgD{\qbarcovn{\xn_{0j}}{\alpha_n}{\beta_n}}{\qbarcovn{\xn_{0j}}{0}{0}} \nex
		& \dupspace +\!\! \sum_\zn \!\bracknorm{\!\qhatn{j}\!\bracknorm{\zn} \!-\! \qbarcovn{\xn_{0j}}{\alpha_n}{\beta_n}\!\!\bracknorm{\zn}\!} \!\log \!\frac{\qbarcovn{\xn_{0j}}{\alpha_n}{\beta_n}\!\!\bracknorm{\zn}}{\qbarcovn{\xn_{0j}}{0}{0}\!\bracknorm{\zn}}. \label{eq:98}
	\end{align}
	We bound the absolute value of the last term in~\eqref{eq:98} for a large $n$ by
	\begin{align}
		\abs{\sum_\zn \bracknorm{\qhatn{j} \bracknorm{\zn} - \qbarcovn{\xn_{0j}}{\alpha_n}{\beta_n}\bracknorm{\zn}} \log \frac{\qbarcovn{\xn_{0j}}{\alpha_n}{\beta_n}\bracknorm{\zn}}{\qbarcovn{\xn_{0j}}{0}{0}\bracknorm{\zn}}} \stackrel{(a)}{\leq}  \exp \bracknorm{-\xi_6 n \alpha_n} + \exp \bracknorm{-\xi_6 n \beta_n}, \label{eq:99}
	\end{align}
	for an appropriate $\xi_6 > 0$, where $(a)$ follows from using steps similar to~\cite[(249)-(254)]{ArumugamBloch2018}. 
	Combining~\eqref{eq:97f} to~\eqref{eq:99}, we conclude that for a large $n$, 
	\begin{align}
		 \abs{\avgD{\qhatn{j}}{\qbarn{j}} - \avgD{\qbarcovn{\xn_{0j}}{\alpha_n}{\beta_n}}{\qbarcovn{\xn_{0j}}{0}{0}} }  \leq \exp\bracknorm{-\xi_7 n \alpha_n} + \exp\bracknorm{-\xi_7 n \beta_n}, \label{eq:100}
	\end{align}
	for an appropriate constant $\xi_7 > 0$.
	\paragraph{Asymptotic behavior} We now establish the asymptotic scaling of $\log M_1$ for the proposed covert communication scheme. 
	Combining~\eqref{eq:cov1} and~\eqref{eq:100}, for a fixed $W_2 = j$, we bound $\smash{\avgD{\qhatn{j}}{\qbarn{j}}}$ by
	\begin{align}
		\avgD{\qhatn{j}}{\qbarn{j}} & \leq n \bracknorm{ (1-\lambda^*)\frac{\alpha_n^2}{2} \bracknorm{1 + \sqrt{\alpha_n}} \chisquare{\qo}{\qz}  + \lambda^* \frac{\beta_n^2}{2} \bracknorm{1 + \sqrt{\beta_n}} \chisquare{\qz}{\qo}} \nex
		& \dupspace + \exp\bracknorm{-\xi_7 n \alpha_n} + \exp\bracknorm{-\xi_7 n \beta_n}. \label{eq:128} \\
		\avgD{\qhatn{j}}{\qbarn{j}} & \geq n \bracknorm{ (1-\lambda^*)\frac{\alpha_n^2}{2} \bracknorm{1 - \sqrt{\alpha_n}} \chisquare{\qo}{\qz}  + \lambda^* \frac{\beta_n^2}{2} \bracknorm{1 - \sqrt{\beta_n}} \chisquare{\qz}{\qo}}  \nex
		& \dupspace - \exp\bracknorm{-\xi_7 n \alpha_n} - \exp\bracknorm{-\xi_7 n \beta_n}. \label{eq:129}
	\end{align}
	Ultimately, combining~\eqref{eq:77},~\eqref{eq:81},~\eqref{eq:128}, and~\eqref{eq:129}, and taking the limit as $n \to \infty$, we obtain
	\begin{align} 
		\limn \frac{\log M_1}{\sqrt{n \avgD{\qhatn{j}}{\qbarn{j}}}} & \leq \sqrt{2} \bracknorm{1 - \mu} \frac{\bracknorm{1 - \lambda^*} \avgD{\po}{\pz} + \lambda^* \gamma \avgD{\pz}{\po}}{\sqrt{\bracknorm{1 - \lambda^*} \chisquare{\qo}{\qz} + \lambda^* \gamma^2 \chisquare{\qz}{\qo}}}, \label{eq:130} \\
		\limn \frac{\log M_1}{\sqrt{n \avgD{\qhatn{j}}{\qbarn{j}}}} & \geq \sqrt{2} \bracknorm{1 + \nu} \frac{\bracknorm{1 - \lambda^*} \avgD{\qo}{\qz} + \lambda^* \gamma \avgD{\qz}{\qo}}{\sqrt{\bracknorm{1 - \lambda^*} \chisquare{\qo}{\qz} + \lambda^* \gamma^2 \chisquare{\qz}{\qo}}}. \label{eq:131}	
	\end{align}
\end{proof}

\begin{theorem} \label{thm:con}
	For the channel model described in Section~\ref{sec:channelmodel}, consider a sequence of codes with increasing block length $n$ such that $\smash{\limn P_e^{(1)} = \limn P_e^{(2)} = 0}$, and for all $j \in \intseq{1}{M_2}$, $\smash{\limn \avgD{\qhatn{j}}{\qbarn{j}} = 0}$. 
	If the common message is transmitted using a codebook that achieves the capacity of the channel to Willie, then for every $j \in \intseq{1}{M_2}$, we obtain
	\begin{align}
		\limn \frac{\log M_1}{\sqrt{n \avgD{\qhatn{j}}{\qbarn{j}}}} \leq \max_{\gamma \geq 0} \sqrt{2} \frac{\bracknorm{1 - \lambda^*} \avgD{\po}{\pz} + \lambda^* \gamma \avgD{\pz}{\po}}{\sqrt{\bracknorm{1 - \lambda^*} \chisquare{\qo}{\qz} + \lambda^* \gamma^2 \chisquare{\qz}{\qo}}}. \label{eq:con1}
	\end{align} 
	For a sequence of schemes such that~\eqref{eq:con1} holds with equality and for some $\gamma^*$ that maximizes the right hand side of~\eqref{eq:con1}, we obtain
	\begin{align}
		\limn \frac{\log M_1}{\sqrt{n \avgD{\qhatn{j}}{\qbarn{j}}}} \geq \sqrt{2} \frac{\bracknorm{1 - \lambda^*} \avgD{\qo}{\qz} + \lambda^* \gamma^* \avgD{\qz}{\qo}}{\sqrt{\bracknorm{1 - \lambda^*} \chisquare{\qo}{\qz} + \lambda^* \bracknorm{\gamma^*}^2 \chisquare{\qz}{\qo}}}. \label{eq:con2}
	\end{align} 
\end{theorem}
\begin{proof} \label{prf:con}
	Consider a capacity-achieving codebook $\calC^*$ for the channel between Alice and Willie. 
	For a fixed common message $W_2 = j$, consider a covert communication scheme that is characterized by $\smash{\epsilon_{n,j} \eqdef \P{\widehat{W}_1 \neq W_1 | W_2 = j}}$ and $\delta_{n,j} \eqdef \smash{\avgD{\qhatn{j}}{\qbarn{j}}}$.
	Note that $\limn \epsilon_{n,j} = \limn \delta_{n,j} = 0$. 
	We denote the innocent codeword corresponding to $W_2 = j$ by $\xn_{0j} = \bracknorm{x_{0j,1},x_{0j,2},\ldots, x_{0j,n}}$, the innocent symbol at position $\ell$ by $x_{0j, \ell}$, and the information symbol at position $\ell$ by $x_{0j,\ell}^c \eqdef 1 - x_{0j,\ell}$. 
	For a fixed $W_2 = j$, we denote the input distribution of the sequence $\xn_{0j}$ by $\Pi_j^n$, where the distribution of the symbol at position $\ell$ is defined by
	\begin{align}
		\Pi_{j,\ell}(x_{0j,\ell}^c) = 1 - \Pi_{j,\ell}(x_{0j,\ell}) = \mun{\ell}. \label{eq:26}
	\end{align} 
	We interpret $\mun{\ell}$ as the probability of flipping innocent symbol $x_{0j,\ell}$ to information symbol $x_{0j,\ell}^c$ at symbol position $\ell \in \intseq{1}{n}$. 
	Note that the innocent symbol $x_{0j,\ell}$ depends on the choice of the common message $W_2 = j$ and the symbol position $\ell \in \intseq{1}{n}$. 
	For conciseness, we define the following terms. 
	\begin{align}
		& \pzc{\ell}(y) \eqdef \wyx(y|x_{0j,\ell}), \quad \poc{\ell}(y) \eqdef \wyx(y|x_{0j,\ell}^c), \label{eq:27} \\ 
		& \qzc{\ell}(z) \eqdef \wzx(z|x_{0j,\ell}), \quad \qoc{\ell}(z) \eqdef \wzx(z|x_{0j,\ell}^c). \label{eq:28}
	\end{align}
	Define $K_{j,\ell}(z) \eqdef \qoc{\ell}(z) - \qzc{\ell}(z)$.
	Note that $\forall z \in \calZ$, $K_{j,\ell}(z)$ equals either $\qo(z) - \qz(z)$ or $\qz(z) - \qo(z)$ depending on the choices of $j$ and $\ell$. 
	Defining $K(z) \eqdef \smash{\abs{\qoc{\ell}(z) - \qzc{\ell}(z)}}$, we remove the dependency of $K(z)$ on $j$ and $\ell$. 
	We then define the distribution of each symbol $Z_\ell$ of $\bzn$ by $\qhatj{\ell}$, where
	\begin{align}
		\qhatj{\ell}(z) & \eqdef \sum_x \Pi_{j,\ell}(x) \wzx(z|x) \\
		& = \qzc{\ell}(z) + \mun{\ell} K_{j,\ell}(z). \label{eq:29}
	\end{align}
	Let us now analyze the KL divergence between $\qhatn{j}$ and $\qbarn{j}$. 
	\begin{align}
		\!\!\!\!\delta_{n,j} \!& = \sum_{\zn} \qhatn{j}(\zn) \log \frac{\qhatn{j}(\zn)}{\qbarn{j}(\zn)} \\
		& =  -\avgH{\bzn|W_2 = j} - \sum_{\zn} \qhatn{j}(\zn) \log \qbarn{j}(\zn) \\ 
		& = \!\sum_{\ell = 1}^n \! \bracknorm{\! -\avgH{Z_\ell|\bzn_1^{\ell-1}\!, W_2 \!=\! j} \!-\! \sum_z \! \qhatj{\ell}(z) \!\log \qzc{\ell}(z) \!} \\
		& \geq \sum_{\ell = 1}^n \bracknorm{ \!-\avgH{Z_\ell|W_2 = j} - \sum_z \qhatj{\ell(z)} \log \qzc{\ell}(z)\!} \\ 
		& = \sum_{\ell = 1}^n \avgD{\qhatj{\ell}}{\qzc{\ell}} \label{eq:30}
	\end{align}
	Since $\limn \delta_{n,j} = 0$ and KL divergence terms are non-negative, we obtain $\limn \smash{\avgD{\qhatj{\ell}}{\qzc{\ell}} }= 0 $ for all $\ell \in \intseq{1}{n}$. 
	Using Pinsker's inequality, we obtain $\limn  \V{\qhatj{\ell}, \qzc{\ell}} = 0 $, which implies that $\forall z \in \calZ$ and $\forall \ell \in \intseq{1}{n}$, 
	\begin{align}
		\limn \abs{\qhatj{\ell}(z) - \qzc{\ell}(z)} = 0, \\
		\limn \mun{\ell} K(z) = 0. \label{eq:31}
	\end{align}
	However, there exists at least one $z \in \calZ$ such that $K(z) \neq 0$. Hence, we obtain $\limn \mun{\ell} = 0$.
	Next, define
	\begin{align}
		\Psij{\ell}(z) & \eqdef \mun{\ell} K_{j, \ell}(z), \label{eq:32} \\
		\xij{\ell}(z) & \eqdef \frac{\Psij{\ell}(z)}{\qzc{\ell}(z)} + \frac{4}{3} \frac{\abs{\Psij{\ell}(z)}}{\qzc{\ell}(z)}, \label{eq:33} \\
		\xi_j^{(n)}(z) & \eqdef \max_{\ell \in \intseq{1}{n}} \xij{\ell}(z) . \label{eq:34}
	\end{align}
	Since $\limn \mun{\ell} = 0$, we infer that $\limn \Psij{\ell}(z) = \limn \xij{\ell}(z) = 0$ for all $z \in \calZ$ and $\ell \in \intseq{1}{n}$. 
	Consequently, $\limn \xi_{j}^{(n)}(z) = 0$, $\forall z \in \calZ$. 
	Continuing the analysis of $\delta_{n,j}$ from~\eqref{eq:30}, we have, for $n$ large enough, 
	\begin{align}
		\delta_{n,j} & \geq \sum_{\ell = 1}^n \sum_z \qhatj{\ell}(z) \log \bracknorm{1 + \frac{\mun{\ell}K_{j, \ell}(z)}{\qzc{\ell}(z)}}  \\
		& \stackrel{(a)}{ \geq } \!\! \sum_{\ell = 1}^n \! \sum_z \!\frac{\bracknorm{\Psij{\ell}(z)}^2}{2 \qzc{\ell}(z)} \!\!\bracknorm{\!\! 1\! -\! \frac{\Psij{\ell}(z)}{\qzc{\ell}(z)} \!-\! \frac{4\abs{\Psij{\ell}(z)}}{3\qzc{\ell}(z)}\!} \displaybreak \\
		& = \sum_{\ell = 1}^n \sum_z \frac{\bracknorm{\Psij{\ell}(z)}^2}{2 \qzc{\ell}(z)} \bracknorm{1 - \xij{\ell}(z)} \\
		& \geq \sum_z \bracknorm{1 - \xi_{j}^{(n)}(z)} \sum_{\ell = 1}^n \frac{\bracknorm{\Psij{\ell}(z)}^2}{2\qzc{\ell}(z)}, \label{eq:35}
	\end{align}
	where $(a)$ follows from the fact that $\log(1+x) \geq x - \frac{x^2}{2}$, for $x \geq 0$, $\log(1+x) \geq x - \frac{x^2}{2} + \frac{2x^3}{3}$, for $x \in \bracksq{ -\frac{1}{2}, 0}$, and $\sum_z \Psij{\ell} = 0 $.  
	Then, we define
	\begin{align}
		\no & \eqdef  \textnormal{wt}\bracknorm{\xn_{0j}} , \quad \nz \eqdef n - \no,  \label{eq:37} \\
		\rhoz & \eqdef \frac{\sum_{\substack{\ell = 1 \\ \ell: x_{0j,\ell} = 0 }}^n \mun{\ell} }{\nz}, \quad \rhoo \eqdef \frac{\sum_{\substack{\ell = 1 \\ \ell: x_{0j,\ell} = 1 }}^n \mun{\ell} }{\no} \label{eq:39},  \\
		\gamma_{j}^{(n)} & \eqdef \frac{\rhoo}{\rhoz}, \quad \lamj \eqdef \frac{\no}{n},  \label{eq:41}
	\end{align}
	where $\nz$ and $\no$ denote the number of $0$'s and $1$'s in $\xn_{0j}$, respectively; $\rhoz$ is the average probability of flipping $0$ to $1$ and $\smash{\rhoo}$ is the average probability of flipping $1$ to $0$. 
	Note that $\smash{\limn \rhoz = \limn \rhoo = 0}$. 
	In addition, we set $\limn \gamma_j^{(n)} = \gamma_j^\dagger  \in \mathbb{R}^+ $.
	If $\limn \gamma_j^{(n)} = 0$ or $\infty$, only symbols in positions with innocent symbol $0$ or $1$, respectively, are used to embed covert information. 
	Else, the sequence $\smash{\{\gamma_j^{(n)}\}}$ is bounded, and we can extract a convergent subsequence with limit $\smash{\gamma_j^\dagger}$. 
	Note that $1 - \lamj = \smash{\frac{\nz}{n}}$. 
	Using Cauchy-Schwarz inequality, we obtain
	\begin{align}
		\sum_{\substack{\ell = 1 \\ \ell: x_{0j,\ell} = 0}}^n \bracknorm{\mun{\ell}}^2 \geq \frac{1}{\nz} \bracknorm{\sum_{\substack{\ell = 1 \\ \ell: x_{0j,\ell} = 0}}^n  \mun{\ell}}^2, \label{eq:42} \\
		\sum_{\substack{\ell = 1 \\ \ell: x_{0j,\ell} = 1}}^n \bracknorm{\mun{\ell}}^2 \geq \frac{1}{\no} \bracknorm{\sum_{\substack{\ell = 1 \\ \ell: x_{0j,\ell} = 1}}^n  \mun{\ell}}^2. \label{eq:43}		
	\end{align}
	From~\eqref{eq:35}, we continue to bound $\delta_{n,j}$ by
	\begin{align}
		\delta_{n,j} & \geq \sum_z \frac{1}{2} \bracknorm{1 - \xi_{j}^{(n)}(z)} \bracknorm{\sum_{\substack{\ell =1 \\ \ell: x_{0j,\ell}=0}}^n \bracknorm{\mun{\ell}}^2 \frac{K^2(z)}{\qz(z)} + \sum_{\substack{\ell =1 \\ \ell: x_{0j,\ell}=1}}^n \bracknorm{\mun{\ell}}^2 \frac{K^2(z)}{\qo(z)}  } \\
		& \stackrel{(a)}{\geq} \sum_z \frac{1}{2} \bracknorm{1 - \xi_{j}^{(n)}(z)} \bracknorm{ \frac{1}{\nz} \bracknorm{ \sum_{\substack{\ell =1 \\ \ell: x_{0j,\ell}=0}}^n \mun{\ell} }^2 \frac{K^2(z)}{\qz(z)} + \frac{1}{\no} \bracknorm{ \sum_{\substack{\ell =1 \\ \ell: x_{0j,\ell}=1}}^n \mun{\ell} }^2 \frac{K^2(z)}{\qo(z)}  } \displaybreak \\
		& = \sum_z \frac{1}{2} \bracknorm{1 - \xi_{j}^{(n)}(z)} \bracknorm{ \nz \bracknorm{ \rhoz }^2 \frac{K^2(z)}{\qz(z)} + \no \bracknorm{ \rhoo }^2 \frac{K^2(z)}{\qo(z)}  } \\
		& = \sum_z \frac{1}{2} n \bracknorm{\rhoz}^2 \bracknorm{1 - \xi_{j}^{(n)}(z)} \bracknorm{ (1 - \lamj ) \frac{K^2(z)}{\qz(z)} +  \lamj \bracknorm{ \gamma_j^{(n)} }^2 \frac{K^2(z)}{\qo(z)}  }, \label{eq:44}
	\end{align} 
	where $(a)$ follows from~\eqref{eq:42} and~\eqref{eq:43}. 
	
	We pause the analysis of $\delta_n,j$ here and define constant composition sub-codebooks $\calF_k \subset \calC^*$ with type $P_k$ in which $k$ denotes the weight of the codeword. 
	As there are $(n+1)$ different types for sequences $\brackcurl{0,1}^n$, there are at most $\bracknorm{n+1}$ such sub-codebooks.
	Let us recall that the codebook $\calC^*$ is a capacity-achieving codebook for the channel between Alice and Willie, that is, the rate of the common message $R \eqdef \frac{\log M_2}{n} = \mathbb{I}\bracknorm{\Lambda, \wzx} - \delta(n) $ with $\limn \delta(n) = 0$. 
	Let us assume that $\forall k \in \intseq{1}{n}, \frac{\log \abs{\calF_k}}{n} \leq R - \delta $ for any $\delta > 0$. 
	For a $\delta' < \delta$ and $n$ large enough, 
	\begin{align}
		M_2 & = \sum_k \abs{\calF_k} \\
		& \leq \sum_k \exp \bracknorm{n \bracknorm{R-\delta}} \\
		& < \bracknorm{n+1} \exp \bracknorm{n \bracknorm{R-\delta}} \\
		& < \exp \bracknorm{n \bracknorm{R-\delta'}}. \label{eq:58}
	\end{align}
	We note that the assumption $\frac{\log \abs{\calF_k}}{n} \leq R - \delta$ for all $k \in \intseq{1}{n}$ results in a contradiction in~\eqref{eq:58} since  $M_2 = \exp \bracknorm{nR}$. Hence, there exists at least one sub-codebook $\calF_{k^*} $ such that 
	\begin{align}
		\frac{\log \abs{\calF_{k^*}}}{n} > R - \delta \label{eq:59}	
	\end{align}
	and $P_{k^*}(1) = 1 - P_{k^*}(0) = \frac{k^*}{n} \in \bracknorm{0,1}$. 
	Using \cite[Corollary~6.4]{Csiszar2011}, we bound the rate of this sub-codebook for an arbitrary $\upsilon > 0$ by
	\begin{align}
		\frac{\log \abs{\calF_{k^*}}}{n}  < \mathbb{I}\bracknorm{P_{k^*}, \wzx{}} + 2\upsilon. \label{eq:60}
	\end{align}
	Combining~\eqref{eq:59} and~\eqref{eq:60}, we obtain
	\begin{align}
		\mathbb{I}\bracknorm{P_{k^*}, \wzx} > \mathbb{I}\bracknorm{\Lambda, \wzx} - \delta - 2 \upsilon. \label{eq:61}
	\end{align}
	Using the unicity of the capacity-achieving input distribution, the concavity of mutual information and~\eqref{eq:61}, we conclude that the type $P_{k^*}$ is arbitrarily close to $\Lambda$ since $\delta$ and $\upsilon$ are arbitrary.
	Consequently, we replace $\lamj$ with $\lambda^{\dagger} \eqdef \lambda^* - \epsilon$ for an arbitrarily small $\epsilon \in \mathbb{R}$. 

	We then bound $\log M_1$ using standard converse steps. 
	\begin{align}
		\!\!\!\!\!\log M_1 & = \avgH{W_1|W_2 = j} \\
		& \leq \avgI{W_1; \byn|W_2 = j} + \Hb{\epsilon_{n,j}} + \epsilon_{n,j} \log M_1 \\
		& \leq \avgI{W_1 \bxn; \byn|W_2 \! = \! j} + \Hb{\epsilon_{n,j}} + \epsilon_{n,j} \log M_1 \\
		& = \avgI{\bxn; \byn|W_2 = j} + \avgI{W_1; \byn | W_2 =j, \bxn} + \Hb{\epsilon_{n,j}} + \epsilon_{n,j} \log M_1 \\
		& \stackrel{(a)}{=} \avgH{\byn| W_2 = j} - \avgH{\byn|\bxn, W_2 = j}  + \Hb{\epsilon_{n,j}} + \epsilon_{n,j} \log M_1 \\
		& \stackrel{(b)}{\leq} \sum_{\ell = 1}^n \bracknorm{ \avgH{Y_\ell| W_2 = j} - \avgH{Y_\ell| X_\ell , W_2 = j} }  + \Hb{\epsilon_{n,j}} + \epsilon_{n,j} \log M_1 \\
		& = \! \sum_{\ell = 1}^n \avgI{X_\ell; Y_\ell | W_2 =j} \!+\! \Hb{\epsilon_{n,j}} \!+\! \epsilon_{n,j} \log M_1, \label{eq:45}
	\end{align}
	where $(a)$ follows from the fact that $\avgI{W_1; \byn | W_2 =j, \bxn} = 0$, and $(b)$ follows from the fact that conditioning reduces entropy and the memoryless property of the channel $\wyx$. 
	Rearranging the terms in~\eqref{eq:45}, we obtain
	\begin{align}
		\log M_1 \leq \frac{\sum_{\ell=1}^n \avgI{X_\ell; Y_\ell | W_2 = j}  + \Hb{\epsilon_{n,j}}}{1 - \epsilon_{n,j}}. \label{eq:46}
	\end{align}
	Defining $\smash{\phatj{\ell}}$ as the distribution of symbol $Y_\ell$ of $\byn$, we upper bound the mutual information term in~\eqref{eq:46} by
	\begin{align}
		\avgI{X_\ell; Y_\ell |W_2 = j} \!\!	& = \sum_y \bracknorm{ \bracknorm{1 - \mun{\ell}} \pzc{\ell}(y) \log\frac{\pzc{\ell}(y)}{\phatj{\ell}(y)} }  +  \sum_y \bracknorm{ \bracknorm{\mun{\ell}} \poc{\ell}(y) \log\frac{\poc{\ell}(y)}{\phatj{\ell}(y)} } \\
		& =\mun{\ell}\avgD{\poc{\ell}}{\pzc{\ell}}  - \avgD{\phatj{\ell}}{\pzc{\ell}}  \label{eq:46_1} \\
		& \leq \mun{\ell} \avgD{\poc{\ell}}{\pzc{\ell}}. \label{eq:47}
	\end{align}
	Combining~\eqref{eq:46} and~\eqref{eq:47}, we obtain
	\begin{align}
		\log M_1 & \leq \frac{\sum_{\ell=1}^n  \mun{\ell} \avgD{\poc{\ell}}{\pzc{\ell}}  + \Hb{\epsilon_{n,j}}}{1 - \epsilon_{n,j}} \\
		& =\! \frac{ \nz \rhoz \avgD{\po}{\pz} \!+\! \no \rhoo \avgD{\pz}{\po} \!+\! \Hb{\epsilon_{n,j}}}{1 - \epsilon_{n,j}} \\
		& = \! \frac{n \rhoz\! \bracknorm{ \!(1 \!-\! \lambda^\dagger) \avgD{\po}{\pz} \!+\! \lambda^\dagger \gamma_j^{(n)} \avgD{\pz}{\po}\!} \!+\! \Hb{\epsilon_{n,j}\!}}{1 - \epsilon_{n,j}}.  \label{eq:48}
	\end{align}
	Since $\limn \log M_1 = \infty$,~\eqref{eq:48} imposes that $\limn n \rhoz = \infty$. Consequently, from~\eqref{eq:44}, we conclude that $\limn\sqrt{n \delta_{n,j}} = \infty$. 
	Combining~\eqref{eq:44},~\eqref{eq:48}, and the facts that $\limn\sqrt{n \delta_{n,j}} = \infty$ and $\limn \Hb{\epsilon_{n,j}} = 0$, we obtain
	\begin{align}
		\limn \frac{\log M_1}{\sqrt{n \delta_{n,j}}} 	& \leq \limn \frac{ n \rhoz \bracknorm{ (1 - \lambda^\dagger) \avgD{\po}{\pz} + \lambda^\dagger \gamma_j^{(n)} \avgD{\pz}{\po}}}{\bracknorm{1 - \epsilon_{n,j}}\sqrt{n\sum_z \frac{n\bracknorm{\rhoz}^2}{2}\bracknorm{ 1 - \xi_j^{(n)}(z) } \bracknorm{ (1-\lambda^\dagger )  \frac{K^2(z)}{Q_0(z)} + \lambda^\dagger \bracknorm{\gamma_j^{(n)}}^2 \frac{K^2(z)}{Q_1(z)}}}} \\
		& = \sqrt{2} \frac{(1 - \lambda^\dagger) \avgD{\po}{\pz} + \lambda^\dagger \gamma_j^\dagger \avgD{\pz}{\po}}{\sqrt{\bracknorm{1-\lambda^\dagger}\chisquare{Q_1}{Q_0} + \lambda^\dagger \bracknorm{\gamma_j^\dagger}^2 \chisquare{Q_0}{Q_1}}}.  \label{eq:49}
	\end{align} 
  
	Following standard steps, we lower bound $\log M_1$ by
	\begin{align}
		\log M_1 & = \avgH{W_1|W_2 =j} \\
		& \geq \avgI{W_1; \bzn | W_2 =j} \\
		& \stackrel{(a)}{=} \avgH{\bzn|W_2 =j} - \avgH{\bzn|\bxn, W_1, W_2 =j} \\
		& \stackrel{(b)}{\geq} \avgH{\bzn|W_2 =j} - \avgH{\bzn|\bxn, W_2 =j} \\ 
		& = \sum_{\xn} \sum_{\zn} \Pi_{j}^n(\xn) \wzxn(\zn|\xn) \log \frac{\wzxn(\zn|\xn)}{\qhatn{j}(\zn)}  \\ 
		& = \!\! \sum_{\ell=1}^n \! \sum_x\! \sum_z \! \Pi_{j,\ell}(x) \wzx(z|x) \log \!\frac{\wzx(z|x)}{\qzc{\ell}(z)}\! - \!\delta_{n,j} \\
		& \stackrel{(c)}{\geq} \sum_{\ell=1}^n \avgI{X_\ell; Z_\ell | W_2 =j} - \delta_{n,j} , \label{eq:50}
	\end{align}
	where $(a)$ follows from the fact that $\bxn$ is a function of $(W_1, W_2)$, $(b)$ follows from the fact that conditioning reduces entropy, and $(c) $ follows from the fact that $\sum_x \Pi_{j,\ell}(x) \wzx(z|x) = \qhatj{\ell}(z)$ and from the fact that KL divergence is non-negative. 
	Continuing the analysis of $\log M_1$ by expanding the mutual information term, we obtain
	\begin{align}
		\!\log M_1 \! & \geq \sum_{\ell=1}^n  \sum_z  \bracknorm{1 - \mun{\ell}} \qzc{\ell}(z) \log\frac{\qzc{\ell}(z)}{\qhatj{\ell}(z)} + \sum_{\ell=1}^n \sum_z  \bracknorm{\mun{\ell}} \qoc{\ell}(z) \log\frac{\qoc{\ell}(z)}{\qhatj{\ell}(z)}   - \delta_{n,j} \\
		& \geq \!  \sum_{\ell=1}^n \mun{\ell}\avgD{\qoc{\ell}}{\qzc{\ell}}  \!-\! \sum_{\ell=1}^n \avgD{\qhatj{\ell}}{\qzc{\ell}}  \!-\! \delta_{n,j} \\
		& \stackrel{(a)}{\geq} \sum_{\ell=1}^n \mun{\ell}\avgD{\qoc{\ell}}{\qzc{\ell}}  - 2 \delta_{n,j} \\
		& = \! n \rhoz\!\! \bracknorm{ \!(1 \!-\! \lambda^\dagger) \avgD{\qo}{\qz} \!+\! \lambda^\dagger\! \gamma_j^{(n)} \avgD{\qz}{\qo}\!} \!\!-\! 2 \delta_{n,j}, \label{eq:51}
	\end{align}
	where $(a)$ follows from~\eqref{eq:30}. 
	For an arbitrary $\epsilon \in (0,1)$, an  $n$ large enough, and for any sequence of codes such that~\eqref{eq:49} is satisfied with equality,\footnote{We know that there exists at least one such code from Theorem~\ref{thm:ach}.} we have
	\begin{multline}
		\sqrt{2} \frac{(1 - \epsilon ) \bracknorm{ (1 - \lambda^\dagger) \avgD{\po}{\pz} + \lambda^\dagger \gamma_j^\dagger \avgD{\pz}{\po}}}{\sqrt{\bracknorm{1-\lambda^\dagger}\chisquare{Q_1}{Q_0} + \lambda^\dagger \bracknorm{\gamma_j^\dagger}^2 \chisquare{Q_0}{Q_1}}} \\  \leq \frac{ n \rhoz \bracknorm{ (1 - \lambda^\dagger) \avgD{\po}{\pz} + \lambda^\dagger \gamma_j^{(n)} \avgD{\pz}{\po}}+\Hb{\epsilon_{n,j}} }{\bracknorm{1 - \epsilon_{n,j}} \sqrt{n \delta_{n,j}} }. \label{eq:54}
	\end{multline}
	Combining~\eqref{eq:51}, and~\eqref{eq:54}, we obtain
	\begin{align}
		\frac{\log M_1}{\sqrt{n \delta_{n,j}}}  & \!\geq\! \frac{ n\rhoz\! \bracknorm{\!\bracknorm{1 - \lambda^\dagger} \!\avgD{\qo}{\qz} \!+\! \lambda^\dagger \!\gamma_j^{(n)} \avgD{\qz}{\qo} \!} }{ \sqrt{\!\bracknorm{1\!-\!\lambda^\dagger}\!\chisquare{Q_1}{Q_0} \!+\! \lambda^\dagger \!\bracknorm{\!\gamma_j^\dagger\!}^2\!\!\! \chisquare{Q_0}{Q_1}}} \nex
		& \dupspace \times \frac{\sqrt{2}\! \bracknorm{1 \!-\! \epsilon_{n,j}}\! \bracknorm{1 - \epsilon} \!\bracknorm{ \!(1 \!-\! \lambda^\dagger) \avgD{\po}{\pz} \!+\! \lambda^\dagger \!\gamma_j^\dagger \avgD{\pz}{\po}\!}}{ n \rhoz\! \bracknorm{\! (1 \!-\! \lambda^\dagger) \avgD{\po}{\pz} \!+\! \lambda^\dagger  \gamma_j^{(n)} \avgD{\pz}{\po}\!}\!+\!\Hb{\epsilon_{n,j}\!} } \!-\! \frac{2\delta_{n,j}}{\sqrt{n \delta_{n,j}}}. \label{eq:55}
	\end{align} 
	Since $\limn \delta_{n,j} = 0$, the last term in ~\eqref{eq:55} vanishes in the limit. 
		Since $\epsilon$ is arbitrary, on applying limits to~\eqref{eq:55}, we obtain
	\begin{align}
		\!\limn\! \frac{\log M_1}{\sqrt{n \delta_{n,j}}} \!\geq\! \frac{ \sqrt{2} \bracknorm{\!\bracknorm{1 - \lambda^\dagger} \!\avgD{\qo}{\qz} \!+\! \lambda^\dagger \gamma_j^\dagger \avgD{\qz}{\qo} \!} }{ \sqrt{\!\bracknorm{1\!-\!\lambda^\dagger}\!\chisquare{Q_1}{Q_0} \!+\! \lambda^\dagger \!\bracknorm{\!\gamma_j^\dagger\!}^2\!\!\! \chisquare{Q_0}{Q_1}}} . \label{eq:57}
	\end{align}
	Note that the bounds~\eqref{eq:49} and~\eqref{eq:57} still depend on the choice of the common message $W_2 = j$ through $\gamma_j^{\dagger}$. 
	To eliminate this dependency, we choose an optimal $\gamma^* \geq 0$ that maximizes~\eqref{eq:49} provided the following condition is satisfied. 
	\begin{align}
		(1 - \lambda^\dagger) \avgD{\po}{\pz} + \lambda^\dagger \gamma^* \avgD{\pz}{\po} \geq (1 - \lambda^\dagger) \avgD{\qo}{\qz} + \lambda^\dagger \gamma^* \avgD{\qz}{\qo}.  \label{eq:64}
	\end{align}
	Consequently, replacing $\gamma_j^\dagger$ with $\gamma^*$ in~\eqref{eq:49} and~\eqref{eq:57}, we obtain~\eqref{eq:con1} and~\eqref{eq:con2} since $\epsilon$ in the definition of $\lambda^\dagger$ is arbitrary. 
%
\end{proof}

The combination of~\eqref{eq:con1} and~\eqref{eq:con2} imposes
\begin{align}
	\bracknorm{1 - \lambda^*} \avgD{\po}{\pz} + \lambda^* \gamma \avgD{\pz}{\po} >	\bracknorm{1 - \lambda^*} \avgD{\qo}{\qz} + \lambda^* \gamma \avgD{\qz}{\qo}, \label{eq:18}
\end{align}
which characterizes the advantage that Bob should possess over Willie to facilitate keyless embedding of covert bits. 
Although we normalize $\log M_1$ by $\sqrt{n \avgD{\qhatn{j}}{\qbarn{j}}}$, which depends on the choice of the common message $W_2 = j$, the bounds on $\log M_1$ are independent of $j$. 
\begin{remark}
	As a special case, let us assume that the channel to Willie is degraded \ac{wrt} the channel to Bob. 
	This assumption guarantees that~\eqref{eq:18} is satisfied as degradedness implies $\avgD{\po}{\pz} > \avgD{\qo}{\qz}$ and $\avgD{\pz}{\po} > \avgD{\qz}{\qo}$. 
\end{remark}
However, the degraded broadcast channel assumption is not a necessary condition to facilitate keyless covert communication in our case as shown in the following example. 
\begin{example}
Let us consider a discrete memoryless channel $\bracknorm{\calX, \wzx, \calZ}$	with $\avgD{\po}{\pz} < \avgD{\qo}{\qz}$ and $\avgD{\pz}{\po} > \avgD{\qz}{\qo}$.
Here, the channel $\wzx$ is not degraded \ac{wrt} the channel $\wyx$. Since all KL divergence terms and $\lambda^*$ in~\eqref{eq:18} are constants determined by the channel, the only degree of freedom is $\gamma$. 
\end{example}
\begin{remark}
	For symmetric channels with two inputs, note that $\avgD{\po}{\pz} = \avgD{\pz}{\po}$, $\avgD{Q_1}{Q_0}=\avgD{Q_0}{Q_1}$, and $\chisquare{Q_1}{Q_0} = \chisquare{Q_0}{Q_1}$. 
	Consequently, from~\eqref{eq:con1} and~\eqref{eq:con2}, we obtain
	\begin{align}
		\!\limn\! \frac{\log M_1}{\!\!\sqrt{n \avgD{\qhatn{j}}{\qbarn{j}}}} & \!\leq\! \sqrt{2} \frac{\avgD{\po}{\pz}}{\sqrt{\chisquare{\qo}{\qz}}}, \label{eq:132} \\
		\!\limn\! \frac{\log M_1}{\!\!\sqrt{n \avgD{\qhatn{j}}{\qbarn{j}}}} & \!\geq\! \sqrt{2} \frac{\avgD{\qo}{\qz}}{\sqrt{\chisquare{\qo}{\qz}}}, \label{eq:133}
	\end{align}
	since $\gamma^*=1$ and $\lambda^* = \frac{1}{2}$. 
	Note that the covert throughput in~\eqref{eq:132}  matches that of the point-to-point channel~\cite{Bloch2016}. 
	As a special case, we consider a broadcast setup for \acp{BSC} with $p_B$ and $p_W$ as the crossover probabilities for the channels from Alice to Bob and Willie, respectively. 
	Assuming $p_B \leq 0.5$ and $p_W \leq 0.5$ without loss of generality, we obtain  
	\begin{align}
		& \avgD{\po}{\pz} = \avgD{\pz}{\po} = \bracknorm{1 - 2 p_B} \log \bracknorm{ \frac{1-p_B}{p_B} }, \label{eq:19} \\
		& \avgD{Q_1}{Q_0} \!=\! \avgD{Q_0}{Q_1} \!=\! \bracknorm{1 \!-\! 2 p_W} \log \bracknorm{ \frac{1\!-\!p_W}{p_W} }, \label{eq:20} \\
		& \chisquare{Q_1}{Q_0} = \chisquare{Q_0}{Q_1} = \frac{\bracknorm{1 - 2p_W}^2}{p_W \bracknorm{1-p_W}}. \label{eq:22}
	\end{align}
	Combining~\eqref{eq:132} to \eqref{eq:22}, we obtain 
	\begin{align}
		\!\limn\! \frac{\log M_1}{\!\!\sqrt{n \avgD{\qhatn{j}}{\qbarn{j}}}} & \!\leq\! \!\sqrt{2 p_W\!\bracknorm{1-p_W}} \frac{1 \!-\! 2 p_B}{1 \!-\! 2 p_W} \!\log\! \bracknorm{\! \frac{1\!-\!p_B}{p_B} \!},\label{eq:23}  \\
		\!\limn\! \frac{\log M_1}{\!\!\sqrt{n \avgD{\qhatn{j}}{\qbarn{j}}}} & \!\geq\!  \sqrt{2 p_W \bracknorm{1-p_W}} \log \bracknorm{ \frac{1-p_W}{p_W} }. \label{eq:24}
	\end{align}
	Note that keyless covert communication is achievable in this channel model iff
	\begin{align}
		\bracknorm{1 \!-\! 2 p_B} \log \bracknorm{\! \frac{1\!-\!p_B}{p_B} \!} & \!\geq\! \bracknorm{1 \!-\! 2 p_W} \log \bracknorm{\! \frac{1\!-\!p_W}{p_W} \!}. \label{eq:25}
	\end{align}
\end{remark}

\appendices
\section{Proof of Lemma~\ref{lem:covertprocess}} \label{sec:covprf}
	Since $\qbarcovn{\xbarn}{\alpha_n}{\beta_n}$ is an $n$-fold distribution, we write $\qbarcovn{\xbarn}{\alpha_n}{\beta_n} = \prod_{i=1}^n \qbarcov{\xbar_i}{\alpha_n}{\beta_n} $. 
	We now analyze the KL divergence between $\qbarcovn{\xbarn}{\alpha_n}{\beta_n}$ and $\qbarcovn{\xbarn}{0}{0}$. 
	\begin{align}
		\avgD{\qbarcovn{\xbarn}{\alpha_n}{\beta_n}}{\qbarcovn{\xbarn}{0}{0}} & =\sum_{i=1}^n \avgD{\qbarcov{\xbar_i}{\alpha_n}{\beta_n}}{W_{Z|X = \xbar_i}} \\
		& = n \bracknorm{1 - \lambda_n} \avgD{\qbarcov{0}{\alpha_n}{\beta_n}}{\qz} +  n \lambda_n \avgD{\qbarcov{1}{\alpha_n}{\beta_n}}{\qo}. \label{eq:11}
	\end{align}
	For $k \in \mathbb{N}^*$ and two distributions defined on the same alphabet $\calZ$, we define $\chi_k\bracknorm{P\|Q} \eqdef \sum_z \frac{\bracknorm{P(z) - Q(z)}^k}{Q^{k-1}(z)}$ and $\eta_k\bracknorm{P\|Q} \eqdef \sum_{z: P(z) - Q(z)<0} \frac{\bracknorm{P(z) - Q(z)}^k}{Q^{k-1}(z)}$.
	Then, using~\cite[Lemma 1]{Bloch2016}, we upper bound each of the two KL divergence terms in~\eqref{eq:11} by
	\begin{align}
		\avgD{\qbarcov{0}{\alpha_n}{\beta_n}}{\qz} & \leq \frac{\alpha_n^2}{2} \chisquare{\qo}{\qz} - \frac{\alpha_n^3}{6} \chi_3\bracknorm{\qo \|\qz} + \frac{\alpha_n^4}{3} \chi_4\bracknorm{\qo \|\qz}, \label{eq:12} \\
		\avgD{\qbarcov{1}{\alpha_n}{\beta_n}}{\qo} & \leq \frac{\beta_n^2}{2} \chisquare{\qz}{\qo} - \frac{\beta_n^3}{6} \chi_3\bracknorm{\qz \|\qo} + \frac{\beta_n^4}{3} \chi_4\bracknorm{\qz \|\qo}, \label{eq:13}
	\end{align}
	For $n$ large enough, using~\cite[Lemma 1]{Bloch2016}, we lower bound the two KL divergence terms in~\eqref{eq:11} by
	\begin{align}
		\avgD{\qbarcov{0}{\alpha_n}{\beta_n}}{\qz} & \geq \frac{\alpha_n^2}{2} \chisquare{\qo}{\qz} - \alpha_n^3 \bracknorm{\frac{1}{2}\chi_3\bracknorm{\qo \|\qz} - \frac{2}{3} \eta_3\bracknorm{\qo \| \qz} } + \frac{2\alpha_n^4}{3} \eta_4\bracknorm{\qo \|\qz}, \label{eq:14} \\
		\avgD{\qbarcov{1}{\alpha_n}{\beta_n}}{\qo} & \geq \frac{\beta_n^2}{2} \chisquare{\qz}{\qo} - \beta_n^3 \bracknorm{\frac{1}{2}\chi_3\bracknorm{\qz \|\qo} - \frac{2}{3} \eta_3\bracknorm{\qz \| \qo} } + \frac{2\beta_n^4}{3} \eta_4\bracknorm{\qz \|\qo}. \label{eq:15}
	\end{align}
	Loosening the bounds in~\eqref{eq:12}-\eqref{eq:15}, for $n$ large enough, we obtain
	\begin{align}
		\frac{\alpha_n^2}{2}\bracknorm{1 + \sqrt{\alpha_n}} \chisquare{\qo}{\qz} & \geq \avgD{\qbarcov{0}{\alpha_n}{\beta_n}}{\qz} \geq \frac{\alpha_n^2}{2}\bracknorm{1 - \sqrt{\alpha_n}} \chisquare{\qo}{\qz}, \label{eq:16} \\
		\frac{\beta_n^2}{2}\bracknorm{1 + \sqrt{\beta_n}} \chisquare{\qz}{\qo} & \geq \avgD{\qbarcov{1}{\alpha_n}{\beta_n}}{\qo} \geq \frac{\beta_n^2}{2}\bracknorm{1 - \sqrt{\beta_n}} \chisquare{\qz}{\qo}. \label{eq:17}
	\end{align}
	Ultimately, combining~\eqref{eq:11},~\eqref{eq:16}, and~\eqref{eq:17}, we obtain~\eqref{eq:cov1}. 

\section{Proof of Lemma~\ref{lem:rel}} \label{sec:relprf}
	We denote the covert transmission status of Alice by $T \eqdef 1 - \indic{W_1 = 0}$ and Bob's estimate of $T$ by $\widehat{T}$. 
	For $j \in \intseq{1}{M_2}$ and $\xn_{0j} \in \calC_2$, define
	\begin{align}
		\!\!\! \agammaj \eqdef \brackcurl{\bracknorm{\xn, \yn} \in \calX^n \times \calY^n : \log \frac{\wyxn\bracknorm{\yn, \xn}}{\pbarcovn{\xn_{0j}}{\alpha_n}{\beta_n}(\yn)} \geq \gamma_j }, \label{eq:72}
	\end{align}
	where $\gamma_j > 0$ will be determined later. 
	The decoder at Bob operates as follows
	\begin{itemize}
		\item if $\exists$ unique $i$ such that $\bracknorm{\xn_{ij}, \yn} \in \agammaj$, output $\widehat{W}_1 = i$,
		\item else if $\not\exists$ $i$ such that $\bracknorm{\xn_{ij}, \yn} \in \agammaj$, output $\widehat{W}_1 = 0$, 
		\item else, declare a decoding error. 
	\end{itemize}
		Define the event $E_{ij} \eqdef \smash{\brackcurl{\bracknorm{\bxn_{ij}, \byn} \in \agammaj} }$. 
		We also define
		\begin{align}
			\!\!\!\!E_1 & \!\eqdef\! \mathbb{E}_{\calC}\!\!\bracknorm{\! \frac{1}{M_2}\! \sum_{j=1}^{M_2} \!\sum_{t' \in \brackcurl{0,1}}\!\!\!\! \P{\!\widehat{T}\!\neq\! t'\middle|t \!=\! t', \widehat{W}_2 \!=\! W_2 \!=\! j }\!\!}, \\
			\!\!\!\!\!\!\!\!\!E_2 & \!\eqdef\! \mathbb{E}_{\calC}\!\!\bracknorm{\frac{1}{M_2}\! \sum_{j=1}^{M_2}\P{\! \widehat{W}_1 \!\neq\! W_1 \middle| \widehat{W}_2 \!=\! W_2 \!=\! j, \widehat{T}\!=\!t\!=\!1}}. 
		\end{align}
		The error probability of the covert message averaged over all choices of the  codebook $\calC$ can be written as
	\begin{align}
		\!\mathbb{E}_{\calC}\!\!\bracknorm{\frac{1}{M_2}\! \sum_{j=1}^{M_2}\!P_{e,2,j}^{\bracknorm{1}}\!} & =  E_1 + E_2. \label{eq:101}
	\end{align}
	From the definition of $E_1$, we obtain
	\begin{align}
		E_1 &\!=\! \E{\calC}{\frac{1}{M_2} \sum_{j=1}^{M_2}\P{\widehat{T} \!=\! 0 \middle|t\!=\!1, \widehat{W}_2 \!=\! W_2\!=\! j }\!}  \!\!\!+\! \E{\calC}{\!\frac{1}{M_2} \!\sum_{j=1}^{M_2}\P{\!\widehat{T} \!=\! 1 \middle|t\!=\!0, \widehat{W}_2 \!=\!W_2 \!=\! j }\!}. \label{eq:102}
	\end{align}
	We upper bound the first term in~\eqref{eq:102} by
	\begin{align}
		\E{\calC}{\frac{1}{M_2} \sum_{j=1}^{M_2}\P{\widehat{T} \!=\! 0 \middle|t\!=\!1, \widehat{W}_2 \!=\! W_2\!=\! j }\!} & = \E{\calC}{ \!\!\frac{1}{M_1M_2}\! \sum_{j=1}^{M_2} \!\sum_{i=1}^{M_1} \!\sum_\yn \! \wyxn\!\bracknorm{\yn|\bxn_{ij}} \!\indic{\bigcap_{i'} E_{i'j}^c \!} \!}   \\
		& \stackrel{(a)}{\leq} \E{\calC}{ \frac{1}{M_1M_2} \sum_{j=1}^{M_2} \sum_{i=1}^{M_1} \sum_\yn \wyxn\bracknorm{\yn|\bxn_{ij}} \indic{E_{ij}^c } } \label{eq:103_1} \\
		& = \frac{1}{M_2} \sum_{j=1}^{M_2} \sum_{\xn_{0j}} P_X^n\bracknorm{\xn_{0j}} \mathbb{P}_{\wyxn\pibar{\xn_{0j}}{\alpha_n}{\beta_n}}\bracknorm{\bracknorm{\agammaj}^c}, \label{eq:103}
	\end{align}
	where $(a)$ follows from the fact that the probability of intersection of several events does not exceed the probability of one of those events. We bound the second term in~\eqref{eq:102} by
	\begin{align}
		 \!\!\!\!&  \E{\calC}{\!\frac{1}{M_2} \!\sum_{j=1}^{M_2}\P{\!\widehat{T} \!=\! 1 \middle|t\!=\!0, \widehat{W}_2 \!=\!W_2 \!=\! j }\!} \nex 
		 & = \E{\calC}{\frac{1}{M_2} \!\sum_{j=1}^{M_2}\sum_\yn \wyxn\bracknorm{\yn|\bxn_{0j}} \indic{\bigcup_{i} E_{ij} }} \\ 
		 & \stackrel{(a)}{\leq} \! \frac{1}{M_2} \!\sum_{j=1}^{M_2} \! \sum_{\xn_{0j}} P_X^n\bracknorm{\xn_{0j}} \sum_{i=1}^{M_1}\! \sum_\yn\! \sum_{\xn_{ij}} \!\wyxn\!\!\bracknorm{\yn|\xn_{0j}\!} \!  \pibar{\xn_{0j}}{\alpha_n}{\beta_n}\!\!\bracknorm{\xn_{ij}} \!\indic{\!\bracknorm{\xn_{ij}, \yn}\!\! \in\! \agammaj\!} \\
		 & \leq \frac{1}{M_2} \!\sum_{j=1}^{M_2} \! \sum_{\xn_{0j}} P_X^n\bracknorm{\xn_{0j}}  M_1 e^{-\gamma_j}  \!\! \sum_\yn\! \sum_{\xn_{1j}} \!\frac{\wyxn\!\bracknorm{\yn|\xn_{1j}}}{\pbarcovn{\xn_{0j}}{\alpha_n}{\beta_n}\!\bracknorm{\yn}}  \wyxn\bracknorm{\yn|\xn_{0j}} \pibar{\xn_{0j}}{\alpha_n}{\beta_n}\bracknorm{\xn_{1j}} \\
		 & \stackrel{(b)}{\leq} \frac{1}{M_2} \!\sum_{j=1}^{M_2} M_1 e^{-\gamma_j},  \label{eq:104}
	\end{align}
	where $(a)$ follows from the union bound and $(b)$ follows from the fact that $\sum_{\xn_{1j}} \wyxn\bracknorm{\yn|\xn_{1j}} \pibar{\xn_{0j}}{\alpha_n}{\beta_n}\bracknorm{\xn_{1j}} = \pbarcovn{\xn_{0j}}{\alpha_n}{\beta_n}\bracknorm{\yn}$ and the definition of $\agammaj$. 
	We then bound the second term in~\eqref{eq:101} by
	\begin{align} 
		& \!\!\!\!E_2 \stackrel{(a)}{\leq} \E{\calC}{ \frac{1}{M_1M_2} \sum_{j=1}^{M_2}\sum_{i=1}^{M_1} \sum_\yn \wyxn\bracknorm{\yn|\bxn_{ij}} \indic{E_{ij}^c} } \nex
		& \dupspace \! +\! \mathbb{E}_{\calC}\!\!\bracknorm{\!\! \frac{1}{M_1M_2} \sum_{j=1}^{M_2}\!\sum_{i=1}^{M_1} \!\sum_{\substack{i'=1 \\ i'\neq i}}^{M_1}\!\sum_\yn \!\wyxn\!\bracknorm{\yn|\bxn_{ij}}\! \indic{E_{i'j} } \!\!}, \label{eq:106}
	\end{align}
	where $(a)$ follows from the union bound. 
	We upper bound the second term in~\eqref{eq:106} by
	\begin{align}
		& \mathbb{E}_{\calC}\!\!\bracknorm{\!\! \frac{1}{M_1M_2} \sum_{j=1}^{M_2}\!\sum_{i=1}^{M_1} \!\sum_{\substack{i'=1 \\ i'\neq i}}^{M_1}\!\sum_\yn \!\wyxn\!\bracknorm{\yn|\bxn_{ij}}\! \indic{E_{i'j} } \!\!} \nex
		& \leq\! \frac{M_1}{M_2} \sum_{j=1}^{M_2} \sum_{\xn_{0j}} P_X^n\bracknorm{\xn_{0j}} \!\sum_\yn\! \sum_{\xn_{1j}} \pbarcovn{\xn_{0j}}{\alpha_n}{\beta_n}\!\!\bracknorm{\yn}  \!\pibar{\xn_{0j}}{\alpha_n}{\beta_n}\bracknorm{\xn_{1j}} \indic{\bracknorm{\xn_{1j}, \yn} \in \agammaj } \\
		& \leq \frac{1}{M_2} \sum_{j=1}^{M_2} M_1 e^{-\gamma_j}. \label{eq:108}
	\end{align}
	Define $\gamma_j \eqdef \bracknorm{1-\delta} \sum_{i=1}^n \avgI{X_i; Y_i | \overline{X}_i = x_{0j,i}} $ for an arbitrary $\delta \in \bracknorm{0,1}$.
	Expanding $\avgI{X_i; Y_i | \overline{X}_i = x_{0j,i}}$, we obtain
	\begin{align}
		& \avgI{X_i;Y_i|\overline{X}_i = x_{0j,i}} \nex
		& = \bracknorm{\sum_y \bracknorm{1-\alpha_n} \pz(y) \log \frac{\pz(y)}{\pbarcov{0}{\alpha_n}{\beta_n}(y)}  + \sum_y \alpha_n \po(y) \log \frac{\po(y)}{\pbarcov{0}{\alpha_n}{\beta_n}(y)} } \indic{x_{0j,i} = 0} \nex
		& \dupspace + \bracknorm{\sum_y \beta_n \pz(y) \log \frac{\pz(y)}{\pbarcov{1}{\alpha_n}{\beta_n}(y)}  + \sum_y \bracknorm{1-\beta_n} \po(y) \log \frac{\po(y)}{\pbarcov{1}{\alpha_n}{\beta_n}(y)} } \indic{x_{0j,i} = 1} \\
		& = \bracknorm{\alpha_n \avgD{\po}{\pz} \!-\!\avgD{\pbarcov{0}{\alpha_n}{\beta_n}}{\pz} } \!\indic{x_{0j,i} \!=\! 0} \!+\!  \bracknorm{\beta_n \avgD{\pz}{\po} \!-\! \avgD{\pbarcov{1}{\alpha_n}{\beta_n}}{\po} } \indic{x_{0j,i} \!=\! 1} \\
		& \stackrel{(a)}{=} \bracknorm{\alpha_n \avgD{\po}{\pz} + \bigO{\alpha_n^2} } \indic{x_{0j,i} = 0} +  \bracknorm{\beta_n \avgD{\pz}{\po} + \bigO{\beta_n^2} } \indic{x_{0j,i} = 1}. \label{eq:115}
	\end{align}
	where $(a)$ follows from combining~\eqref{eq:12},~\eqref{eq:13},~\eqref{eq:14}, and~\eqref{eq:15}, in the proof of Lemma~\ref{lem:covertprocess}. 
	Aggregating the $n$ mutual information terms corresponding to each symbol position, we obtain
	\begin{align}
		\sum_{i=1}^n \avgI{X_i;Y_i|\overline{X}_i = x_{0j,i}} & = n \bracknorm{\bracknorm{1 - \lambda_j}\alpha_n \avgD{\po}{\pz} + \lambda_j \beta_n \avgD{\pz}{\po} } + n \bigO{\alpha_n^2} + n \bigO{\beta_n^2}. \label{eq:117}
	\end{align}
	We bound the probability term in~\eqref{eq:103} by  
	\begin{align}
		\!\mathbb{P}_{\wyxn\pibar{\xn_{0j}}{\alpha_n}{\beta_n} }\!\!\bracknorm{\!\bracknorm{\agammaj\!}^c}\! & \stackrel{(a)}{\leq} \!\exp\bracknorm{-\zeta_1 n \alpha_n} \!+\! \exp \bracknorm{-\zeta_1 n \beta_n}, \label{eq:118}
	\end{align}
	for an appropriate $\zeta_1 > 0$, where $(a)$ follows from using Bernstein's inequality as in~\cite[Appendix D]{ArumugamBloch2018}. 
	Then, combining~\eqref{eq:101},~\eqref{eq:102},~\eqref{eq:103_1},~\eqref{eq:103},~\eqref{eq:104},~\eqref{eq:106},~\eqref{eq:108}, and~\eqref{eq:118}, we infer that~\eqref{eq:78} is satisfied for a large $n$ and appropriate constant $\xi_2 > 0$ if, for every $j \in \intseq{1}{M_2}$, $\log M_1$ satisfies
	\begin{align}
		\log M_1 < \bracknorm{1 - \delta} n \bracknorm{\bracknorm{1 - \lambda_j}\alpha_n \avgD{\po}{\pz} + \lambda_j \beta_n \avgD{\pz}{\po} } + n \bigO{\alpha_n^2} + n \bigO{\beta_n^2}. \label{eq:117_1}
	\end{align}
	However, since $\lambda_j$ is arbitrarily close to $\lambda^*$ and $\delta$ is arbitrary, it is sufficient if $\log M_1$ satisfies~\eqref{eq:77}.
\section{Proof of Lemma~\ref{lem:res}} \label{sec:resprf}
	For $W_2 = j$ and $\xn_{0j} \in \calC_2$, define the set
	\begin{align}
		\!\!\!\!\btauj \eqdef \brackcurl{\bracknorm{\xn, \zn} \in \calX^n \times \calZ^n : \log \frac{\wzxn{}\bracknorm{\zn|\xn}}{\qbarcovn{\xn_{0j}}{\alpha_n}{\beta_n}\bracknorm{\zn}} \leq \tau_j }, \label{eq:120}
	\end{align}
	where $\tau_j > 0$ will be determined later. 
	For a fixed $j$ and $i \in \intseq{1}{M_1}$, the expectation over all random codewords $\brackcurl{\bxn_{kj}}_{k \in \intseq{1}{M_1}\backslash\brackcurl{i}}$ is denoted by $\mathbb{E}_{\sim i}$.  
	We bound the KL divergence between $\smash{\qhatn{W_2}}$ and $\smash{\qbarcovn{\xn_{0W_2}}{\alpha_n}{\beta_n}}$ averaged over all choices of the common message $W_2$ and the codebook by
	\begin{align}
		& \E{\calC}{\mathbb{E}_{W_2}\avgD{\qhatn{W_2}}{\qbarcovn{\bxn_{0W_2}}{\alpha_n}{\beta_n}}} \nex
		& \dupspace = \E{\calC}{\frac{1}{M_2} \sum_{j=1}^{M_2} \sum_\zn \frac{1}{M_1} \sum_{i=1}^{M_1} \wzxn{}\bracknorm{\zn|\bxn_{ij}} \log \frac{\sum_{k=1}^{M_1} \wzxn{}\bracknorm{\zn|\bxn_{kj}}}{M_1 \qbarcovn{\bxn_{0j}}{\alpha_n}{\beta_n}\bracknorm{\zn}} } \\
		& \dupspace \!=\! \frac{1}{M_1M_2} \!\sum_{j=1}^{M_2} \!\sum_{i=1}^{M_1}\! \sum_{\xn_{0j}}\! P_X^n \bracknorm{\xn_{0j}} \!\sum_\zn \!\sum_{\xn_{ij}} \!\wzxn{}\!\bracknorm{\zn|\xn_{ij}} \!\pibar{\xn_{0j}}{\alpha_n}{\beta_n}\!\! \bracknorm{\xn_{ij}} \nex
		& \dupspace \dupspace \times \!\E{\sim i}{ \!\log\!\! \bracknorm{\! \frac{\sum_{\substack{k=1\\k\neq i}}^{M_1}\! \wzxn{}\!\bracknorm{\zn|\bxn_{kj}}}{M_1 \qbarcovn{\xn_{0j}}{\alpha_n}{\beta_n}\bracknorm{\zn}} \!+\! \frac{\wzxn{}\bracknorm{\zn|\xn_{ij}}}{M_1 \!\qbarcovn{\xn_{0j}}{\alpha_n}{\beta_n}\!\bracknorm{\zn}} \!}\!\! }\\
		& \dupspace \! \stackrel{(a)}{\leq} \! \frac{1}{M_1M_2} \!\sum_{j=1}^{M_2} \!\sum_{i=1}^{M_1}\! \sum_{\xn_{0j}}\! P_X^n \bracknorm{\xn_{0j}} \!\sum_\zn \!\sum_{\xn_{ij}} \!\wzxn{}\!\bracknorm{\zn|\xn_{ij}} \!\pibar{\xn_{0j}}{\alpha_n}{\beta_n}\!\! \bracknorm{\xn_{ij}} \nex
		& \dupspace \dupspace \times \log \E{\sim i}{ \! \frac{\sum_{\substack{k=1\\k\neq i}}^{M_1}\! \wzxn{}\!\bracknorm{\zn|\bxn_{kj}}}{M_1 \qbarcovn{\xn_{0j}}{\alpha_n}{\beta_n}\bracknorm{\zn}} \!+\! \frac{\wzxn{}\bracknorm{\zn|\xn_{ij}}}{M_1 \!\qbarcovn{\xn_{0j}}{\alpha_n}{\beta_n}\!\bracknorm{\zn}} \!} \label{eq:121}, 
	\end{align}
where $(a)$ follows from Jensen's inequality.	
Defining $\smash{\mumin} \eqdef \min\brackcurl{\bracknorm{1 - \alpha_n}\min_z Q_0(z), \bracknorm{1 - \beta_n} \min_z Q_1(z)}$, we bound the log term in~\eqref{eq:121} by
	\begin{align}
		& \!\!\log \mathbb{E}_{\sim i}\!\bracknorm{\! \frac{\sum_{\substack{k=1\\k\neq i}}^{M_1}\! \wzxn{}\!\bracknorm{\zn|\bxn_{kj}}}{M_1 \qbarcovn{\xn_{0j}}{\alpha_n}{\beta_n}\bracknorm{\zn}} \!+\! \frac{\wzxn{}\bracknorm{\zn|\xn_{ij}}}{M_1 \qbarcovn{\xn_{0j}}{\alpha_n}{\beta_n}\!\!\bracknorm{\zn}} \!} \nex
		& \dupspace \!\!=\! \log \!\bracknorm{\! \frac{\sum_{\substack{k=1\\k\neq i}}^{M_1}\! \sum_{\xn_{kj}}\!\! \wzxn{}\!\bracknorm{\zn|\xn_{kj}} \pibar{\xn_{0j}}{\alpha_n}{\beta_n}\!\bracknorm{\xn_{kj}}\! }{M_1 \qbarcovn{\xn_{0j}}{\alpha_n}{\beta_n}\!\!\bracknorm{\zn}\!} \!+\! \frac{\wzxn{}\bracknorm{\zn|\xn_{ij}}}{M_1 \qbarcovn{\xn_{0j}}{\alpha_n}{\beta_n}\bracknorm{\zn}} \!} \\
		& \dupspace \stackrel{(a)}{\leq} \log \bracknorm{1 + \frac{\wzxn{}\bracknorm{\zn|\xn_{ij}}}{M_1 \qbarcovn{\xn_{0j}}{\alpha_n}{\beta_n}\bracknorm{\zn}}} \\ 
		& \dupspace \leq \log \bracknorm{1 + \frac{e^{\tau_j}}{M_1}} + \log \bracknorm{1 + \frac{1}{\qbarcovn{\xn_{0j}}{\alpha_n}{\beta_n} \bracknorm{\zn}}} \indic{\bracknorm{\xn_{ij}, \zn} \not \in \btauj} \\
		& \dupspace \leq \frac{e^{\tau_j}}{M_1} + n \log \bracknorm{\frac{2}{\mumin}} \indic{\bracknorm{\xn_{ij}, \zn} \not \in \btauj}\label{eq:122}
	\end{align}
	where $(a)$ follows from the fact that $\sum_{\xn_{kj}} \wzxn{}\bracknorm{\zn|\xn_{kj}} \pibar{\xn_{0j}}{\alpha_n}{\beta_n} \bracknorm{\xn_{kj}} = \qbarcovn{\xn_{0j}}{\alpha_n}{\beta_n}\bracknorm{\zn}$. 
	Combining~\eqref{eq:121} and~\eqref{eq:122}, we obtain
	\begin{align}
		\E{\calC}{\mathbb{E}_{W_2}\avgD{\qhatn{W_2}}{\qbarcovn{\bxn_{0W_2}}{\alpha_n}{\beta_n}}}  & \leq n \log \bracknorm{\frac{2}{\mumin}} \frac{1}{M_2} \sum_{j=1}^{M_2} \sum_{\xn_{0j}}\! P_X^n \bracknorm{\xn_{0j}} \P[\wzxn\pibar{\xn_{0j}}{\alpha_n}{\beta_n} ]{\bracknorm{\btauj}^c} \nex
		& \dupspace + \frac{1}{M_2} \sum_{j=1}^{M_2} \frac{e^{\tau_j}}{M_1} . \label{eq:125}
	\end{align}
Using steps similar to those used to obtain~\eqref{eq:117} in Appendix~\ref{sec:relprf}, we obtain
	\begin{align}
		\sum_{i=1}^n \mathbb{I}\bracknorm{X_i;Z_i|\overline{X}_i = x_{0j,i}} & = n \bracknorm{\bracknorm{1 - \lambda_j}\alpha_n \avgD{\qo}{\qz} + \lambda_j \beta_n \avgD{\qz}{\qo} } + n \bigO{\alpha_n^2} + n \bigO{\beta_n^2} \label{eq:126}
	\end{align} 
	Defining $\tau_j \eqdef \bracknorm{1+\delta} \sum_{i=1}^n \avgI{X_i;Z_i|\overline{X}_i = x_{0j,i}} $ for an arbitrary $\delta> 0$, we bound the probability term on the right hand side of~\eqref{eq:125} using Bernstein's inequality by
	\begin{align}
		\!\P[\wzxn\pibar{\xn_{0j}}{\alpha_n}{\beta_n} ]{\bracknorm{\btauj}^c} \!\leq \!\exp\!\bracknorm{-\zeta_2 n \alpha_n} \!+\! \exp \!\bracknorm{-\zeta_2 n \beta_n}\!,\label{eq:127}
	\end{align}
	for an appropriate $\zeta_2>0$. 
	Consequently, combining~\eqref{eq:125} and~\eqref{eq:127} and ensuring\footnote{Similar to Appendix~\ref{sec:relprf}, we remove the dependency of $\log M_1$ on $j$ via $\lambda_j$ in $\tau_j$ using the fact that $\lambda_j$ is arbitrarily close to $\lambda^*$ for every $j \in \intseq{1}{M_2}$.} $\log M_1$ satisfies~\eqref{eq:81} for an arbitrary $\nu \in \bracknorm{0,1}$ and a large $n$,
	we conclude that there exists a constant $ \xi_3 > 0$ such that~\eqref{eq:82} is satisfied. 

\bibliographystyle{IEEEtranS}

\end{document}